\documentclass{IEEEoj}

\usepackage{cite}

\usepackage{amsmath,epsfig,amssymb,verbatim,amsopn,cite,multirow}
\usepackage{amsthm}
\usepackage{balance}
\usepackage{multirow}
\usepackage{graphicx}
\usepackage{subcaption}

\usepackage{epstopdf}
\usepackage{amsmath,amssymb,amsfonts}
\usepackage{algorithm,algorithmic}
\usepackage{graphicx,color}
\usepackage{textcomp}
\def\BibTeX{{\rm B\kern-.05em{\sc i\kern-.025em b}\kern-.08em
    T\kern-.1667em\lower.7ex\hbox{E}\kern-.125emX}}
\AtBeginDocument{\definecolor{ojcolor}{cmyk}{0.93,0.59,0.15,0.02}}
\def\OJlogo{\vspace{-14pt}\includegraphics[height=24pt]{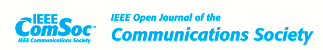}}
\begin{document}
\receiveddate{XX Month, XXXX}
\reviseddate{XX Month, XXXX}
\accepteddate{XX Month, XXXX}
\publisheddate{XX Month, XXXX}
\currentdate{XX Month, XXXX}
\doiinfo{OJCOMS.2022.1234567}

\title{Protecting Massive MIMO-Radar Coexistence: Precoding Design and Power Control}

\author{Mohamed Elfiatoure\authorrefmark{1}, Mohammadali Mohammadi\authorrefmark{1},  ~\IEEEmembership{Senior Member,~IEEE,}\\ Hien Quoc Ngo\authorrefmark{1}, \IEEEmembership{Senior Member,~IEEE,} Peter J. Smith\authorrefmark{2}, \IEEEmembership{Fellow,~IEEE,}\\  Michail Matthaiou\authorrefmark{1},   \IEEEmembership{Fellow,~IEEE}
\affil{Centre for Wireless Innovation (CWI), Queen's University Belfast, U.K.}
\affil{School of Mathematics and Statistics, Victoria University of Wellington, Wellington 6012, New Zealand}
\corresp{CORRESPONDING AUTHOR: Mohamed Elfiatoure (e-mail: melfiatoure01@qub.ac.uk).}
\authornote{This work is a contribution by Project REASON, a UK Government
funded project under the Future Open Networks Research Challenge (FONRC)
sponsored by the Department of Science Innovation and Technology (DSIT). The work
of M. Mohamamdi and M. Matthaiou  was supported  by the European Research Council
(ERC) under the European Union’s Horizon 2020 research
and innovation programme (grant agreement No. 101001331).The work of H. Q. Ngo was supported by the U.K. Research and Innovation Future Leaders Fellowships under Grant MR/X010635/1.}
\markboth{Preparation of Papers for IEEE OPEN JOURNALS}{Author \textit{et al.}}}

\newtheorem{Lemma}{Lemma}
\newtheorem{Theorem}{Theorem}
\newtheorem{Corollary}[Lemma]{Corollary}
\newtheorem{Definition}[Lemma]{Definition}
\newtheorem{Algorithm}[Lemma]{Algorithm}
\newtheorem{Remark}{Remark}
\newtheorem{Conjecture}[Lemma]{Conjecture}
\newtheorem{Assumption}[Lemma]{Assumption}
\newtheorem{Example}[Lemma]{Example}
\newenvironment{Proof}
  {\proof}{\proofend}
\newtheorem{proposition}{Proposition}
\newtheorem{property}{Property}

\newcommand{\com}[1]{\textcolor{red}{#1}}
\newcommand{\comm}[1]{\textcolor{blue}{#1}}

\def\cf{\emph{c.f}\onedot} \def\Cf{\emph{C.f}\onedot}
\def\etc{\emph{etc}\onedot}
\def\wrt{w.r.t\onedot} \def \dof{d.o.f\onedot}
\def\etal{\emph{et al}\onedot}

\newcommand{\bigformulatop}[2]{%
\begin{figure*}[!t]
\normalsize
\setcounter{mytempeqcounter}{\value{equation}}
\setcounter{equation}{#1}
#2

\setcounter{equation}{\value{mytempeqcounter}}
\hrulefill
\vspace*{4pt}
\end{figure*}
}
\newcommand{\ve}[1]{\boldsymbol{#1}}
\newcommand{\E}[1]{E\left\{#1\right\}}
\newcommand{\vA}{\ve{A}} \newcommand{\va}{\ve{a}}
\newcommand{\vB}{\ve{B}} \newcommand{\vb}{\ve{b}}
\newcommand{\vC}{\ve{C}} \newcommand{\vc}{\ve{c}}
\newcommand{\vD}{\ve{D}} \newcommand{\vd}{\ve{d}}
\newcommand{\vE}{\ve{E}} \newcommand{\vEe}{\ve{e}}
\newcommand{\vF}{\ve{F}} \newcommand{\vf}{\ve{f}}
\newcommand{\vG}{\ve{G}} \newcommand{\vg}{\ve{g}}
\newcommand{\vH}{\ve{H}} \newcommand{\vh}{\ve{h}}
\newcommand{\vI}{\ve{I}} \newcommand{\vi}{\ve{i}}
\newcommand{\vJ}{\ve{J}} \newcommand{\vj}{\ve{j}}
\newcommand{\vK}{\ve{K}} \newcommand{\vk}{\ve{k}}
\newcommand{\vL}{\ve{L}} \newcommand{\vl}{\ve{l}}
\newcommand{\vM}{\ve{M}} \newcommand{\vm}{\ve{m}}
\newcommand{\vN}{\ve{N}} \newcommand{\vn}{\ve{n}}
\newcommand{\vO}{\ve{O}} \newcommand{\vo}{\ve{o}}
\newcommand{\vP}{\ve{P}} \newcommand{\vp}{\ve{p}}
\newcommand{\vQ}{\ve{Q}} \newcommand{\vq}{\ve{q}}
\newcommand{\vR}{\ve{R}} \newcommand{\vr}{\ve{r}}
\newcommand{\vS}{\ve{S}} \newcommand{\vs}{\ve{s}}
\newcommand{\vT}{\ve{T}} \newcommand{\vt}{\ve{t}}
\newcommand{\vU}{\ve{U}} \newcommand{\vu}{\ve{u}}
\newcommand{\vV}{\ve{V}} \newcommand{\vv}{\ve{v}}
\newcommand{\vW}{\ve{W}} \newcommand{\vw}{\ve{w}}
\newcommand{\vX}{\ve{X}} \newcommand{\vx}{\ve{x}}
\newcommand{\vY}{\ve{Y}} \newcommand{\vy}{\ve{y}}
\newcommand{\vZ}{\ve{Z}} \newcommand{\vz}{\ve{z}}

\newcommand{\qa}{{\bf a}}
\newcommand{\qb}{{\bf b}}
\newcommand{\qc}{{\bf c}}
\newcommand{\qd}{{\bf d}}
\newcommand{\qe}{{\bf e}}
\newcommand{\qf}{{\bf f}}
\newcommand{\qg}{{\bf g}}
\newcommand{\qh}{{\bf h}}
\newcommand{\qi}{{\bf i}}
\newcommand{\qj}{{\bf j}}
\newcommand{\qk}{{\bf k}}
\newcommand{\ql}{{\bf l}}
\newcommand{\qm}{{\bf m}}
\newcommand{\qn}{{\bf n}}
\newcommand{\qo}{{\bf o}}
\newcommand{\qp}{{\bf p}}
\newcommand{\qq}{{\bf q}}
\newcommand{\qr}{{\bf r}}
\newcommand{\qs}{{\bf s}}
\newcommand{\qt}{{\bf t}}
\newcommand{\qu}{{\bf u}}
\newcommand{\qv}{{\bf v}}
\newcommand{\qw}{{\bf w}}
\newcommand{\qx}{{\bf x}}
\newcommand{\qy}{{\bf y}}
\newcommand{\qz}{{\bf z}}

\newcommand{\qA}{{\bf A}}
\newcommand{\qB}{{\bf B}}
\newcommand{\qC}{{\bf C}}
\newcommand{\qD}{{\bf D}}
\newcommand{\qE}{{\bf E}}
\newcommand{\qF}{{\bf F}}
\newcommand{\qG}{{\bf G}}
\newcommand{\qH}{{\bf H}}
\newcommand{\qI}{{\bf I}}
\newcommand{\qJ}{{\bf J}}
\newcommand{\qK}{{\bf K}}
\newcommand{\qL}{{\bf L}}
\newcommand{\qM}{{\bf M}}
\newcommand{\qN}{{\bf N}}
\newcommand{\qO}{{\bf O}}
\newcommand{\qP}{{\bf P}}
\newcommand{\qQ}{{\bf Q}}
\newcommand{\qR}{{\bf R}}
\newcommand{\qS}{{\bf S}}
\newcommand{\qT}{{\bf T}}
\newcommand{\qU}{{\bf U}}
\newcommand{\qV}{{\bf V}}
\newcommand{\qW}{{\bf W}}
\newcommand{\qX}{{\bf X}}
\newcommand{\qY}{{\bf Y}}
\newcommand{\qZ}{{\bf Z}}

\newcommand{\ykp}{\breve{\qy}_{k,\mathrm{p}}}
\newcommand{\ykpp}{\breve{\qy}_{k',\mathrm{p}}}

\newcommand{\taur}{\tau_{\mathrm {r}}}
\newcommand{\PZF}{\mathrm{PZF}}
\newcommand{\ZF}{\mathrm{ZF}}
\newcommand{\MR}{\mathrm{MR}}
\newcommand{\SINR}{\mathrm{SINR}}
\newcommand{\SE}{\mathrm{SE}}
\newcommand{\cor}{\mathrm{Cor}}
\newcommand{\hR}{\hat{\qR}}

\newcommand{\onevec}{\boldsymbol{1}}

\newcommand{\new}{{\mathsf{new}}}

\newcommand{\vect}{\text{vec}}

\newcommand{\vyy}[1]{\ve{{y}_{#1}}}
\newcommand{\dv}[1]{\boldsymbol{#1}}

\newcommand{\Rx}{\mathrm{Rx}}
\newcommand{\Tx}{\mathrm{Tx}}
\newcommand{\powd}{\mathcal{P}}

\newcommand{\Ex}{\mathbb{E}}
\newcommand{\con}[1]{{#1}^{\ast}}
\newcommand{\conn}[1]{{#1}^{\ddag}}
\newcommand{\mct}[1]{{#1}^{\dagger}}
\newcommand{\mt}[1]{{#1}^{T}}
\newcommand{\RxA}{n_R}
\newcommand{\TxA}{n_T}
\newcommand{\vHa}{\vH_{\aleph}}
\newcommand{\PRmax}{P_{R,\mathrm{max}}}
\newcommand{\diag}{\mathrm{diag}}

\newcommand{\Ts}{T_s}
\newcommand{\Prob}{\textnormal{Pr}}
\newcommand{\Prs}[1]{\,\textnormal{Pr}\!\left(#1\right)}
\newcommand{\Prv}[1]{\,\mathrm{Pr}\!\left[#1\right]}
\newcommand{\Muti}{\mathcal{I}}

\newcommand{\Pd}{\bar{\rm {P}}_{d}}

\newcommand{\ettall}{\emph{et al.}}

\newcommand{\AuthorOne}{Mohamed Elfiatoure}
\newcommand{\AuthorTwo}{Hien Quoc Ngo}
\newcommand{\AuthorThree}{and Michail Matthaiou}

\newcommand{\SEk}{\mathrm{SE}_{k}}
\newcommand{\SEth}{\mathrm{SE_{th}}}
\newcommand{\SEkth}{\mathrm{SE}_{k,\mathrm{th}}}
\newcommand{\gamkth}{\gamma_{k,\mathrm{th}}}
\newcommand{\SNRk}{\mathrm{SINR}_{k}}
\newcommand{\trac}{\mathrm{tr}}

\newcommand{\BLETA}{\boldsymbol{\eta}}



\begin{abstract}
This paper studies the coexistence between a downlink multiuser massive multi-input-multi-output (MIMO) communication system and MIMO radar. The performance of the massive MIMO system with maximum ratio ($\MR$), zero-forcing ($\ZF$), and protective $\ZF$ ($\PZF$) precoding designs is characterized in terms of spectral efficiency (SE) and by taking the channel estimation errors and power control into account. The idea of $\PZF$ precoding relies on the projection of the information-bearing signal onto the null space of the radar channel to protect the radar against communication signals. We further derive closed-form expressions for the  detection probability of the radar system for the considered precoding designs. By leveraging the closed-form expressions for the SE and  detection probability, we formulate a power control problem at the radar and base station (BS) to maximize the  detection probability while satisfying the per-user SE requirements. This optimization problem can be efficiently tackled via the bisection method by solving a linear feasibility problem. Our analysis and simulations show that the $\PZF$ design has the highest  detection probability performance among all designs, with intermediate SE performance compared to the other two designs. Moreover, by optimally selecting the power control coefficients at the BS and radar, the detection probability improves significantly.
\end{abstract}



\maketitle

\section{INTRODUCTION}
Spectrum sharing between the radar and cellular communication systems, termed as \emph{communication/radar co-existing systems}, has been envisioned as an enabling solution to address the explosive growth of wireless traffic demands and shortage of licensed spectra~\cite{zheng2019radar,Chiriyath:TSP:2016,liu2020joint,chiriyath2017radar}.
Nevertheless, the inherent challenge of spectrum sharing, i.e., inter-system interference that compromises the performance of both systems, calls for an efficient cross interference management and spectrum assignment. To facilitate co-existence with overlaid communication systems, a variety of techniques, such as opportunistic spectrum sharing between cellular and rotating radar~\cite{saruthirathanaworakun2012opportunistic}, interference mitigation~\cite{zheng2019radar,Deng:AES:2013}, precoding or spatial separation~\cite{zheng2019radar, Khawar:2014}, and waveform design for radar~\cite{zheng2019radar, Aubry:AES:2014,zheng2017joint} have been proposed in the literature. 

The widespread deployment of the massive multiple-input, multiple-output (MIMO) technology in cellular networks on one hand~\cite{Matthaiou:COMMag:2021, zhang2020prospective}, and the potential of MIMO radars on the other hand~\cite{li2007MSP}, has paved the way to the co-existence of MIMO structures.  MIMO technology offers waveform diversity and higher detection capability for the radars, while at the same time,  massive MIMO technology delivers the significant user coverage and spectral efficiency (SE) enhancement for the cellular systems.  While the literature has focused more on conventional MIMO for both systems, the potential of  massive MIMO technology to further boost the system's performance has not been thoroughly studied yet. Massive MIMO is a key enabling technology for 5G and beyond networks, which relies on a large number of antennas at the BS to provide high spectral and energy efficiency using relatively simple processing~\cite{marzetta2016fundamentals}. More importantly, a BS with a large-antenna array can easily form a null to minimize interference to a coexistent radar. This motivates us to develop a massive MIMO communication system overlaid with a radar system. 

\subsection{Related Works}
To manage the interference between the MIMO radar system and MIMO cellular networks, the null space projection method~\cite{Mahal:AES:2017,Biswas:TWC:2018}  and optimum beamforming design~\cite{liu2017robust,Qian:TSP:2018,Liu:TSP:2018,Pu:CLET:2022} have been widely discussed in the literature. 
More specifically, Mahal~\ettall~\cite{Mahal:AES:2017}  proposed a radar precoder design using subspace projection methods and  based on zero-forcing ($\ZF$) and minimum mean-square-error (MMSE) criteria. Biswas~\ettall~\cite{Biswas:TWC:2018} applied null-space based waveform projection to mitigate the interference from the radar system toward a full-duplex cellular system and proposed a joint transceiver design at the base station (BS) and users to maximize the  detection probability of the MIMO radar system. The authors in~\cite{liu2017robust} considered the transmit beamforming design for spectrum sharing between downlink multiuser MIMO communication and colocated MIMO radar to maximize the detection probability of the radar, while guaranteeing the transmit power budget of the BS and the received signal-to-interference-plus-noise-ratio (SINR) of each downlink user. The authors in~\cite{liu2017robust} focused on the transmit beamforming design for spectrum sharing between downlink multiuser MIMO communication and colocated MIMO radar. Their aim was to maximize the radar's detection probability, while ensuring that the base station's transmit power budget and each downlink user's received signal-to-interference-plus-noise ratio (SINR) are maintained. Qian~\ettall~\cite{Qian:TSP:2018} addressed the problem of joint design of the radar transmit code, radar receive filter, and the communication system codebook for the co-existence of MIMO radar and MIMO communication. Liu~\ettall~\cite{Liu:TSP:2018} investigated power efficient transmission in the overlaid systems, where the BS beamforming is designed to minimize the transmit power at the BS, while guaranteeing the receive SINR at the users and the interference level from BS to radar. Pu~\ettall~\cite{Pu:CLET:2022} extended the proposed design in~\cite{Liu:TSP:2018} by taking the radar transmit waveform design into consideration and maximized the radar SINR, under the constraints of
communication constructive interference, radar waveform similarity and constant modulus.

\subsection{Research Gap and Main Contributions}
The integration of the massive MIMO into communication/radar co-existing systems has been recently studied in~\cite{mishra2023mimo,Elfiatoure:JCIN:2023}. In~\cite{mishra2023mimo}, the rate region of  a coexistence based joint radar and communications system, comprising a single cell massive MIMO communication system and a static MIMO radar has been characterized. In the presence of radar interference, the uplink and downlink achievable rates of the cellular system have been derived by applying MMSE combining and regularized $\ZF$ beamformer at the BS, respectively. In~\cite{Elfiatoure:JCIN:2023}, closed-form expressions for the  detection probability of the radar system and the downlink SE of the massive MIMO system with maximum ratio ($\MR$) precoding were derived. Nonetheless, research on the coexistence of massive MIMO communications and radar systems is still in its infancy and how joint precoding design and power allocation affects its performance remains unclear.  

Thanks to the promising features of the massive MIMO technology, in this paper, we investigate the potential benefits of massive MIMO coexisting with MIMO radar to alleviate the inter-system interference. We consider a joint radar communication system comprising a single-cell massive MIMO cellular communication system and a MIMO radar operating over the same frequency band. The specific contributions of our paper can be summarized as follows:

\begin{itemize}
\item We characterize the performance of the cellular and radar system in terms of downlink SE and  detection probability, respectively, in the presence of imperfect channel state information (CSI). Analytical results for $\MR$ and $\ZF$ beamforming at the BS are derived. 
In order to manage the inter-system interference, we design protective $\ZF$ ($\PZF$) precoding at the BS to ensure the radar's functionality is not significantly impaired by the interference caused by downlink transmission towards users. Accordingly, we characterize the performance of the cellular and radar system with the $\PZF$ scheme.

\item We formulate a power allocation problem with the objective of maximizing the detection probability for MIMO radar, subject to a power budget constraint at the radar and minimum SE requirements at the cellular users. This problem is efficiently solved via the bisection method. Our proposed power allocation strategy provides a significant detection probability gain for all precoding designs compared to the case without power allocation.

\item Our numerical results show that by increasing the number of BS antennas, the gap between the PZF and MR design reduces, while ZF constantly outperforms the MR. Nevertheless, while increasing the number of antennas at the radar results in SE degradation, it can also significantly enhance the detection probability at the radar site. By implementing power control at both the BS and radar, the $\MR$, $\ZF$, and $\PZF$ schemes achieve detection probability improvements of up to $55\%$, $49\%$, and $38\%$, respectively, compared to their baseline values, while the detection probability approaches to $1$.
Finally, our findings demonstrate that the $\PZF$ scheme, when combined with optimal power control, can consistently achieve a probability of detection exceeding $0.8$ with a high probability (greater than $0.6$).  
 
 \end{itemize}

\subsection{Paper Organization and Notation}
The rest of this paper is organized as follows. Section~\ref{sec:sysmodel} describes the system model.  Our performance analysis is pursued in Section~\ref{sec:pfanalysis}. In Section~\ref{sec:power allocation}, we develop the proposed power allocation problem. Simulation results are presented in Section~\ref{sec:numerical}. Finally, Section~\ref{sec:conc} contains concluding remarks.

$\textbf{Notations}$: We denote vectors and matrices by lower-case boldface symbols and upper-case boldface symbols, respectively;  $\Vert\cdot\Vert$ denotes the $l_2$ norm, $\trac(\cdot)$ denotes the matrix trace, $(\cdot)^*$, $(\cdot)^H$ and  $(\cdot)^T$ denote the conjugate, conjugate transpose, and transpose, respectively; ${\mathbf I}_M$ denotes an $ M \times M$ identity matrix; $[\qA]_{(:,k)}$ denotes the $k$-th column of $\qA$; $ \mathcal{CN}\left(\boldsymbol{0},  \mathbf{C}_x \right)$ denotes a circularly symmetric complex Gaussian vector  with zero-mean and covariance matrix $\mathbf{C}_x$; $X\sim\mathcal {CN}(0,\sigma^2)$ denotes a circularly symmetric complex Gaussian random variable (RV) $X$ zero mean and variance $\sigma^2$; $X\sim\mathcal {N}(0,1)$ denotes a real valued Gaussian RV; $\Ex\{\cdot\}$ denotes the statistical expectation; $Q_n(\cdot,\cdot)$ denotes the Marcum $Q$-function of order $n$ defined in~\cite[Eq.(4.60)]{simon2001digital}.  
\section{System Model} ~\label{sec:sysmodel}
We consider a time division duplex (TDD) downlink  massive MIMO communication system coexisting with a MIMO radar system on the same time-frequency resource. Radar is equipped with $N$ transmit and $N$ receive antennas, while the massive MIMO communication system includes  an $M$-antenna BS serving $K$ single-antenna users $(M>K)$. The basic structure of the model is illustrated in Fig.~\ref{fig:system}, where

\begin{itemize}
    \item  $\qg_{k}   \in \mathbb{C}^{M \times 1}$  is the channel vector response between the BS and the $k$-th user. The channel   $\qg_{k}$ is modeled as follows: 
\begin{equation}
 \qg_{k}=\sqrt{\beta_{k}} \mathbf{z}_{k}, 
\label{eqn:equation}
\end{equation}
where  $\mathbf{z}_{k}$   represents the small-scale fading, assuming to include independent and identically distributed (i.i.d.) RVs, i.e.,  $\mathbf{z}_{k}$ $   \sim \mathcal{CN}\left(0,  \qI_M \right)$, while $\beta_k$ represents the large-scale fading.  Denote by $\qG = [ \qg_1 \ldots \qg_K] \in \mathbb{C}^{M \times K}$ the corresponding channel matrix from the BS to all $K$ users.

\item $\qR \in \mathbb{C}^{N \times M}$ represents the channel matrix from the BS to the radar receiver, whose elements are i.i.d. $   \sim  \mathcal{CN}\left(0,\beta_{br} \right)$ RVs, where $\beta_{br}$ is the corresponding large-scale fading coefficient.
\item $\mathbf{{f}}_k \in \mathbb{C}^{N \times 1}$ is the  channel  from the radar transmitter to the $k$-th user. Denote by $\mathbf{{F}} = [ \mathbf{{f}}_1 \ldots \mathbf{{f}}_K] \in \mathbb{C}^{N \times K}$ the corresponding channel matrix from the radar transmitter to all $K$ users. We assume that $\mathbf{{f}}_k \sim \mathcal{CN}\left(0,  \bar{\beta}_k\qI_N \right) $, where $\bar{\beta}_k$ represents the large-scale fading.

In the following subsections, we provide details on the architecture of the two systems.

\end{itemize}

\subsection{Massive MIMO Communication System}
Our focus here is on data transmission over the downlink with TDD operation.  Each transmission frame is divided to  two phases: 1) uplink training phase and 2) downlink data transmission. Relying on the channel estimates obtained in the uplink training phase, different linear processing schemes are applied at the BS to transmit information towards all users.
\begin{figure}
  \includegraphics[width=0.50\textwidth]{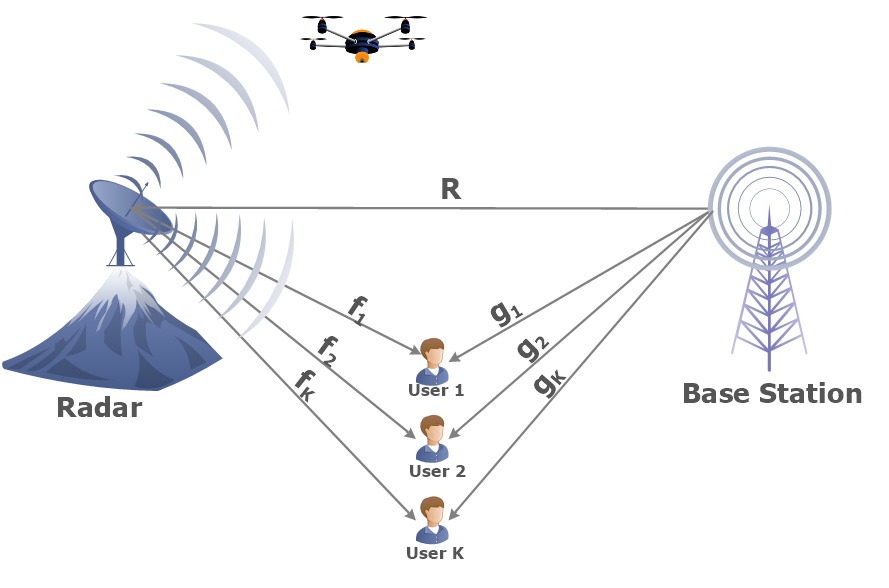}
	\centering
	\caption{Coexistence between massive MIMO cellular systems and MIMO radar.}
   \label{fig:system}
\end{figure}
\vspace{-2em}
\subsubsection{Uplink Training}

In this phase, the users and radar will first send the pilot sequences to the BS. Then, based on the received pilot signals, the BS will estimate the channels to users and radar. The CSI of the radar at the BS is required to design the $\PZF$ scheme and protect the radar against the interference from the BS.
Let $ \tau _{\mathrm {p}} $  be the number of symbols per coherence interval used
for uplink training phase. All users and radar simultaneously transmit
pilot sequences of length $ \tau _{\mathrm {p}} $  symbols. 

Let $\sqrt { \tau _{\mathrm {p}}} \pmb{\Phi}_{r}\in \mathbb{C}^{\tau _{\mathrm {p}} \times N}$, and $\sqrt { \tau _{\mathrm {p}}} \pmb{\Phi}_{p}\in \mathbb{C}^{\tau _{\mathrm {p}} \times K}$
be the pilot sequences sent by the radar and the $K$ users, respectively. It is assumed that pilot sequences transmitted from all users and radar are pairwisely orthogonal, i.e., $\pmb{\Phi}_{r}^{H} \pmb{\Phi}_{r} = \pmb{I}_N$, $\pmb{\Phi}_{p}^{H} \pmb{\Phi}_{p} = \pmb{I}_K$, $\pmb{\Phi}_{r}^{H} \pmb{\Phi}_{p} = \pmb{0}_{N\times K}$, and $\pmb{\Phi}_{p}^{H} \pmb{\Phi}_{r} = \pmb{0}_{K\times N}$. This requires $\tau _{\mathrm {p}} \geq N+K$.
Then, the received pilot signal at the BS can be expressed as
\begin{align} ~\label{eq:Pilotmat}
\qY_{\mathrm{\text p}} = \sqrt { \tau _{\mathrm {p}} P_{R}}\qR^H \pmb{\Phi}_{r}^{H}+\sqrt{\tau_{\mathrm {p}}\rho_{u}}\mathbf{G} \pmb{\Phi}_{p}^{H} + {\mathbf {N}}, 
\end{align}
where $\mathbf{N} \in \mathbb{C}^{M \times \tau _{\mathrm {p}} }$ is the AWGN, whose elements have zero mean and variance $\sigma^2_{C}$.

Denote by $\hat{\qR}$ and  $\hat{\qg}_{k}$ the MMSE estimates of $\qR$ and ${\qg}_{k}$, respectively. Then, they can be expressed as 
\begin{align} \label{eq:error}
\hat{\qR}&=\qR + \tilde{\qR},\\
\hat{\qg}_{k}&=\qg_{k} + \tilde{\qg}_{k},
\end{align}
where $\tilde{\qR}$ and $\tilde{\qg}_{k}$ represent the corresponding channel estimation errors. Following \cite{Kay}, the elements of $\hat{\qR}$ and the elements of  $\hat{\qg}_{k}$  are independent and distributed as
 $\mathcal{CN}(0, \gamma_r)$ and $\mathcal{CN}(0, \gamma_k)$, respectively, with
\begin{align} 
\gamma_{r} 
&=\frac { \tau _{\mathrm {p}} P_{R} \beta _{br}^{2}} 
{ \tau _{\mathrm {p}} P_{R}  \beta _{br}+\sigma^2_{C}},\\
\gamma _{k}&=\frac {{ \tau _{\mathrm {p}} \rho_{u}  }\beta _{k}^{2}}{ \tau _{\mathrm {p}} \rho_{u}  \beta _{k}+\sigma^2_{C}}, 
\end{align}
where  $ P_{R}$ and $\rho_{u}$ are the transmit power of each pilot symbol at the radar and the user, respectively, and $\sigma^2_{C}$ is the noise power at the BS. From the MMSE estimation property, $\tilde{\qR}$ is independent of $\hat{\qR}$ and $\tilde{\qg}_{k} $ is independent of $ \hat{\qg}_{k} $, and hence, $[\tilde{\qR}]_{n,m}\sim\mathcal{CN}\big({0},  (\beta _{br} - \gamma _{r}) \big)$ and $\tilde{\qg}_{k}\sim\mathcal{CN}\big(\boldsymbol{0},  (\beta _{k} - \gamma _{k})\qI_M \big)$.  We can represent the overall channel estimation matrix between the BS and users as
\begin{align}
\hat{\mathbf{G}}&=\qH \mathbf{D}^{\frac{1}{2}}_{\gamma},
\end{align}
where $\hat{\mathbf{G}}=[\hat{\qg}_1,\ldots,\hat{\qg}_K]\in\mathbb{C}^{M\times K}$, $\qH = [\qh_1,\ldots,\qh_K]\in\mathbb{C}^{M\times K}$ with $\qh_{k}\sim\mathcal{CN}(\boldsymbol{0}, \qI_M)$, while $\qD_{\gamma} =\diag\{\gamma_{1},\ldots, \gamma_{K}\}$  is a diagonal matrix. We can further represent $\hat{\qg}_{k}$ as 
\begin{align}~\label{eq:che:gk}
\hat{\qg}_{k}&=\sqrt{\gamma_{k}}{\qh}_{k}.
\end{align}


\vspace{-2em}
\subsubsection{Downlink Signal Transmission}
By using the channel estimates obtained during the training phase, precoding matrix is designed at the BS for downlink data transmission phase.
Suppose that the information symbols $\qd = [d_1,\ldots,d_K]^T\in\mathbb{C}^{K\times 1}$ are independent with  $\Ex\big\{ \qd\qd^H\big\} = \qI_K$, where $d_k$ denotes the information symbol intended for the $k$-th user. Then, the signal transmitted from the BS can be expressed as
\begin{align} \label{eq:x}
&\boldsymbol{x}= \qT \qD^{\frac{1}{2}}_{\eta}\qd
= \sum_{k=1}^{K}   \sqrt{{\eta}_{k}} \qt_{k} {d}_{k},
\end{align}
where  $\qT =[\qt_{1} \ldots \qt_{K} ] \in\mathbb{C}^{M\times K}$ is the precoding matrix, which is a function of the channel estimate  $\hat{\mathbf{G}}$ and $\qt_{k}\in\mathbb{C}^{N\times 1}$, where $\Ex\big\{{\|\qt_{k}\|}^{2}\big\} = 1$, is the precoding vector for user $k$, $\qD_{\eta} =\diag\{\eta_1,\ldots,\eta_K\}$  is a diagonal matrix whose $k$-th diagonal element, $\eta_{k}$, represents the power control coefficient for user $k$, chosen to satisfy the power constraint at the BS $\Ex\{\|\mathbf{x}\|^2\}\leq1$, which implies $\sum_{k=1}^{K} {\eta_{k}} \leq 1$.
The received signal vector at the users can be expressed as
 \begin{align} \label{eq:y1}
\boldsymbol{y} =\underbrace{\sqrt{\rho} \boldsymbol{G}^{\mathrm{T}} \boldsymbol{x}}_\text{desired signal} +\underbrace{\sqrt{P_{R}}\mathbf{F}^{T}\mathbf{s},}_\text{interference from radar}+\boldsymbol{n},
 \end{align}
 where  $\rho$ denotes the BS transmit power, $\qs \in \mathbb{C}^{N \times 1}$  is the transmitted probing signal from the radar, with $\Ex\left\{\mathbf{s} \mathbf{s}^{H}\right\}=\qI_N$, to the target, while $\boldsymbol{n}    \sim C \mathcal{N}\left(0, \sigma _ {C}^{2}\qI_N\right)$ is the AWGN at the users.

 \subsection{ MIMO Radar}
We consider a MIMO radar system  that detects targets located in the far field. Assuming the MIMO radar-to-target channel is line-of-sight  (LoS), the reflected signal (echo from the target) from one point-like target to the radar receiver is interfered by the signal transmitted from the BS. 
Assuming that a uniform linear array (ULA)
is used at the radar, at the $l$-th snapshot, the discrete signal vector $\qy_{R}[l]$ received by the radar is given by~\cite{liu2017robust,ahmed2021reinforcement}
\begin{align}\label{eq:radarRX1}
\qy_{R}[l]=&\alpha \sqrt{P_{R}} \qA(\theta) \qs[l]+\qR \sum_{k=1}^{K}   \sqrt{{\eta}_{k}} \qt_{k} {d}_{k}[l],+\qw[l],
\end{align}
where $\alpha$ denotes the complex path loss of the radar-target-radar path; $P_{R}$ is the transmitted power from the MIMO radar; $\theta$ is the azimuth angle of the target; $\qw[l] = \left[w_i[l],\ldots w_N[l]\right]^T\in\mathbb{C}^{N\times 1}$ is the received additive white Gaussian noise (AWGN) vector at the $l$-th snapshot with $w_m[l] \sim \mathcal{CN}\left(0, \sigma _ {R}^{2}\right)$, $\forall m$; $\qA(\theta)=\qa_R(\theta)\qa_T^T(\theta)$, in which $\qa_T(\theta) \in \mathbb{C}^{N\times 1}$ and $\qa_R(\theta) \in \mathbb{C}^{N\times 1}$ are the transmit and receive steering vectors
of the radar antenna array. Similar to~\cite{liu2017robust,ahmed2021reinforcement}, and without significant loss of generality, we assume that $\qa_R(\theta)=\qa_T(\theta)=\qa(\theta)$, where 
\begin{align}
\qa(\theta)=\left[1, e^{-j {2 \pi d} \sin \left(\theta\right)}, \ldots, e^{-j {2 \pi d}\left(N-1\right) \sin \left(\theta\right)}\right]^{T},
\end{align}
where $d$ represents the inter-antenna spacing normalized by the carrier wavelength.  

\section{Performance Analysis}~\label{sec:pfanalysis}
We evaluate the performance, in terms of downlink SE of the cellular communication systems and  detection probability for the radar system,  for different precoding schemes. Theoretically, the precoding matrix, $\qT$, can be optimized to achieve the optimal performance of the system. However, the complexity of the optimum precoding grows dramatically with $M$ and $K$. For the massive antenna regime with $M\gg K$, it is known that linear precoders, i.e., $\MR$ and $\ZF$ perform fairly well~\cite{Marzetta:TWC:2010,Ngo:TCOM:2013,Khansefid:TCOM:2015}. Therefore, we focus on the performance of those precoders in the following subsections. While these two precoders provide cellular users with satisfactory performance, the BS may create interference at the radar receiver. To address this issue, we propose to use  $\PZF$ precoding at the BS, to guarantee full protection for the
radar against signals intended for cellular users.

\vspace{-0.5em}
\subsection{Spectral Efficiency}
We assume that estimated CSI is not available at the user side and each user uses only the statistical CSI for signal detection. Therefore, users treat the mean effective channel gain as the channel
knowledge for data detection. By invoking~\eqref{eq:y1} and using the use-and-then-forget bounding technique in~\cite[Eq.~(2.44)] {marzetta2016fundamentals}\cite{Hien:cellfree}, known as hardening bound, we derive a lower bound on the downlink SE of user $k$. To this end, using~\eqref{eq:y1} we first rewrite the received signal at user $k$ as
\begin{align}~\label{eq:yi:hardening}
    y_k &=  \mathrm{DS}_k  d_k +
    \mathrm{BU}_k d_k
     +\sum_{k'\neq k}
     \mathrm{IUI}_{kk'}
     d_{k'}
    + \sqrt{P_{R}}\qf^{T}_{k}\qs + n_k,
\end{align}
where 
\vspace{-0.5em}
\begin{align}
 \mathrm{DS}_k  &= \sqrt{{\rho\eta_k} } 
 \Ex\big\{\qg_k^T\qt_k \big\},~\label{eq:DSk}
 \\
 \mathrm{BU}_k  &=  
 \sqrt{{\rho\eta_k} }
 \Big( \qg_k^T\qt_k - \Ex\big\{\qg_k^T\qt_k \big\}\Big),~\label{eq:BUk}
 \\
 \mathrm{IUI}_{kk'} &= 
  \sqrt{{\rho\eta_{k'}}  }
   \qg_k^T\qt_{k'}, ~\label{eq:IUIk}
 \end{align}
represent the strength of the desired signal ($\mathrm{DS}_k$), the beamforming gain uncertainty ($\mathrm{BU}_k$), and the interference caused by the $k'$-th user, respectively.

Accordingly, an achievable downlink SE at the $k$-th user can be expressed as
\begin{align}~\label{eq:dLSE}
\SEk=\Big(1-\frac{\tau_{\mathrm {p}}}{\tau}\Big)\log _{2}(1+\SINR_k),
\end{align}
where the effective SINR is given by
\begin{align}~\label{eq:SINE:general}
    &\SINR_k=\nonumber\\
    &
    \!\frac{
                 \big\vert  \mathrm{DS}_k  \big\vert^2
                 }
                 {  
                 \Ex\Big\{ \big\vert  \mathrm{BU}_k  \big\vert^2\Big\} +
                 \sum_{k'\neq k}
                  \Ex\Big\{ \big\vert \mathrm{IUI}_{kk'} \big\vert^2\Big\}
                  \! +  P_R\bar{\beta}_k N                    +  \sigma_{C}^2}.
\end{align}

The achievable downlink SE in~\eqref{eq:dLSE} is general and valid regardless of the precoding scheme used at the BS. We derive closed-form expressions for the $\MR$, $\ZF$, and $\PZF$ precoding schemes in the following subsection.

\vspace{-0.5em}
\subsubsection{Maximum-Ratio  Precoder}
The $\MR$ precoder is employed at the BS, particularly due to its advantages including low computational complexity, ease of analysis, and reasonable performance, as shown in~\cite{Marzetta:TWC:2010,Ngo:TCOM:2013,Khansefid:TCOM:2015,sutton2021hardening}.  With $\MR$ precoding design, the linear precoding vectors $\qt_k=\qt^{\MR}_k$ given by \cite{Ngo:TCOM:2013} 
\begin{align} \label{e:tmr}
\qt^{\MR}_{k} = \alpha_{\MR}{\qh}^{*}_{k},
\end{align}
where $\alpha_{\MR} = \frac{1}{\sqrt{M}}$ is the normalization factor.
\begin{proposition}\label{prop:DP}
The SE of the $k$-th user achieved by the $\MR$ precoding can be expressed by~\eqref{Prop:SINR:MR}, at the top of the next page.
 \begin{figure*}
\begin{align}  \label{Prop:SINR:MR}
{\SE}^{\MR}_k=\Big(1-\frac{ \tau_{\mathrm {p}}}{\tau}\Big)\log _{2}\left(1+ \frac{M\rho\gamma_{k}\eta_{k}}{\rho \gamma_{k}\sum_{k'=1 }^{K}{\eta_{k'}}+P_R \bar{\beta}_k N+\sigma_{C}^{2}}\right), \forall k.
\end{align}
	\hrulefill
	\vspace{-1mm}
\end{figure*}
\end{proposition}
\begin{proof}
    See Appendix~\ref{app:DP}.
\end{proof}

\vspace{-1.5em}
\subsubsection{Zero-Forcing Precoder}
In order to mitigate the inter-user interference, we exploit the $\ZF$ principle for precoding design. With ZF, the precoding vector $\qt_k$ can be expressed as
\begin{align} \label{eq:tzf}
\qt^{\ZF}_{k} =\alpha_{\ZF}\left [{ \qH}^{*}({\qH}^{T}{\qH}^{*})^{-1}\right ]_{(:,k)},
\end{align}
where $\alpha_{\ZF} = \sqrt{M-K}$ is the normalization factor \cite{ngo2017no}. 

\begin{proposition}~\label{prop:DP:ZF}
The SE of the $k$-th user achieved by the $\ZF$ precoding can be expressed by~\eqref{Prop:SINR:ZF}, at the top of the next page. 
\begin{figure*}
\begin{align}  \label{Prop:SINR:ZF}
{\SE}^{\ZF}_k=\Big(1-\frac{ \tau_{\mathrm {p}}}{\tau}\Big)\log _{2}\left(1+ \frac{(M-K)\rho\gamma_{k}\eta_{k}}{\rho (\beta_{k}-\gamma_{k})\sum_{k'=1}^{K}{\eta_{k'}}+P_R \bar{\beta}_k N+\sigma_{C}^2}\right), \forall k.
\end{align}
	\hrulefill
	\vspace{-1mm}
\end{figure*}
\end{proposition}
\vspace{-0.5em}
\begin{proof}
 See Appendix~\ref{app:ZZFF:proof}.
\end{proof}
\vspace{-1em}
\subsubsection{Protective Zero-Forcing Precoder}
The $\MR$ and $\ZF$ precoders do not take the interference from the BS to radar into account. To protect the radar against the interference caused by the downlink transmission towards cellular users, we now elaborate on the $\PZF$ precoding scheme. The $\PZF$ scheme ensures that the radar is fully protected from BS interference if the BS has perfect CSI from the radar. To achieve this objective, some degrees-of-freedom at the BS are used to steer the BS beams into the radar's null-space. Therefore, from the perspective of the cellular communication network, we can see that $\PZF$ is inferior to $\ZF$, while it performs better than the $\MR$ scheme. Importantly, it outperforms both $\ZF$ and $\MR$ in terms of the detection probability.

Let now $\qB$ denote the projection matrix onto the orthogonal complement of $\hat{\qR}$, which can be expressed as
\begin{align} \label{P:B}
\mathbf {B} = \mathbf {I}_{M}-{\hat{\qR}^{H}}\left ({\hat{\qR}}{\hat{\qR}^{H}}\right)^{-1}{\hat{\qR}}.
\end{align}

Then, the $\PZF$ precoder is designed as
\begin{align}   ~\label{e:tpzf} 
{\mathbf{t}^{\PZF}_{k}}& =\alpha_{\PZF}\mathbf{B}\mathbf{w}_{k}^{\ZF},
\end{align}
 where $\qw_{k}^{\ZF}$ and $\qt^{\PZF}_{k}$ are $k$-th column of  $\qW^{\ZF}$ and $\qT^{\PZF}$ matrices, respectively, where 
$\qW^{\ZF}={\qH}^{*}( \qH^{T}\qH^{*})^{-1}$ and $\qT^{\PZF} = \alpha_{\PZF}\qB \qW^{\ZF}$. Moreover, the normalization factor is expressed as: 
\begin{align}
  \alpha_{\PZF} = 
  \frac{1}{\sqrt{\Ex\{ \|\mathbf{B}\mathbf{w}_{k}^{\ZF}\|^2 \}}}=\frac{1}{\sqrt{\Ex\{\trac( \mathbf{w}_{k}^{\ZF}
(\mathbf{w}_{k}^{\ZF})^{H}\mathbf{B}) \}}},
\end{align}
where we have used $\mathbf{B}^H\mathbf{B}=\mathbf{B}$. By noticing that $\qB$ is independent of $\qw_k^{\ZF}$, we have
\begin{align} \label{P:PZFe}
\alpha_{\PZF} & 
=\frac{1}{\sqrt{\trac\Big( \Ex\big\{\mathbf{w}_{k}^{\ZF}
(\mathbf{w}_{k}^{\ZF})^{H}\big\}\Ex\{\mathbf{B}\}\Big) }}
 \nonumber\\
&=\sqrt{\frac{M(M-K)}{(M-N)}},
\end{align}
where we have applied Lemma~\ref{Lemma:B} from Appendix~\ref{P:lemmaB} and~\cite[Lemma 2.10]{tulino2004random} to derive the final result. 
Accordingly,  the achievable SE by the $k$-th user using $\PZF$ precoding is given in the following proposition.

The idea of PZF is the same as interference nulling which has been introduced in the space of  multi-cell cellular systems~\cite{Tang:TCOM:2013} to mitigate the inter-cell interference and next applied to the cell-free massive MIMO systems in~\cite{Interdonato:TWC:2020} to suppress the interference caused by each access point to part of users.

\begin{proposition}~\label{prop:DP:PZF}
The SE of the $k$-th user achieved by the $\PZF$ precoding can be expressed by~\eqref{Prop:SINR:PZF} at the top of the next page. 
\begin{figure*}
\begin{align}  \label{Prop:SINR:PZF}
{\SE}^{\PZF}_k=\Big(1-\frac{ \tau_{\mathrm {p}}}{\tau}\Big)\log _{2}\left(1+ 
\frac{ \frac{\rho} {M}(M-K)(M-N)\gamma_{k}\eta_{k}}
{\frac{\rho} {M}{(M-N)(M-K)}
 (\beta_{k}-\gamma_{k})\sum_{k'=1}^{K}\eta_{k'}+P_R\bar{\beta}_k  N+\sigma_{C}^2}\right), \forall k.
\end{align}
	\hrulefill
	\vspace{-2mm}
\end{figure*}

\end{proposition}
\begin{proof}
 See Appendix~\ref{app:PZF:proof}.
\end{proof}

\subsection{ Detection Probability}
 In this subsection, we derive the detection probability for radar system, under the Neyman-Pearson criterion. By using the Generalized Likelihood Ratio Test, the asymptotic  detection probability for radar $P_d$ is given as \cite{Elfiatoure:JCIN:2023}
 \begin{align}~\label{eq:detection}
P_{d}=1-\mathfrak{F}_{X_{2}^{2}(\mu)}\left(\mathfrak{F}_{X_{2}^{2}}^{-1}\left(1-P_{F A}\right)\right),
\end{align}
where $P_{FA}$ is the radar’s probability of false alarm, $\mathfrak{F}_{X_{2}^{2}(\mu)}  $   is the non-central chi-square cumulative distribution function (CDF) with two degrees-of-freedom (DoF), $\mathfrak{F}_{X_{2}^{2}}^{-1}$ is the inverse function of chi-square CDF,  while the non-centrality parameter,  $\mu$,  for $ X_2^2$ is given by~\cite{Elfiatoure:JCIN:2023}
\begin{equation}\label{eq:mu1}
\mu=|\alpha|^{2} L P_{R} \trac\left(\qA \qA^{H}\left(\qR\tilde{\qT} \qR^{H}+\sigma_{R}^{2} \qI_N\right)^{-1}\right),
\end{equation}
where $\tilde{\qT} = \qT\qD_{\eta}\qT^H$.

By invoking the generalized Marcum Q-function, $P_{d}$ can be written as
\begin{equation} \label{eq:Tq}
P_{d}=Q_1\big(\sqrt{\mu},\sqrt{C_{FA}}\big),
\end{equation}
where $C_{FA}=\mathfrak{F}_{x_{2}^{2}}^{-1}(1-P_{FA})$~\cite{Elfiatoure:JCIN:2023}. Hence, in order to derive $P_{d}$ for different  precoding schemes, we need to find $\mu$. 

By leveraging the massive MIMO concept, we are able to obtain a tight approximation to $\mu$. Specifically, by using the trace lemma~\cite{Wagner:IT:2012}, as $M$ $\rightarrow \infty$, for the considered linear precoding schemes, we get
\begin{align}
\begin{cases}
  \frac{1}{M}\qR \tilde{\qT}^{i}  \qR^{H} - \frac{\beta_{br}}{M} \trac(\tilde{\qT}^{i} ) \qI_N  \buildrel a.s. \over \rightarrow 0 &\hspace{-4em} {i}\in\{\MR, \ZF\}\\
    \frac{1}{M}\qR \tilde{\qT}^{\PZF}  \qR^{H} - \frac{(\beta_{br} - \gamma_r)}{M} \trac(\tilde{\qT}^{\PZF} ) \qI_N  \buildrel a.s. \over \rightarrow 0, &
  \end{cases}
\end{align}
 where $\buildrel a.s. \over \rightarrow$ denotes the almost sure convergence. Moreover, in the case of $\PZF$ scheme, we have used the fact that
\begin{align}~\label{eq:mu:general}
 \qR \tilde{\qT}^{\PZF}  \qR^{H} &= (\hat{\qR} + \tilde{\qR}) \tilde{\qT}^{\PZF}  (\hat{\qR} + \tilde{\qR})^{H}\nonumber\\
 &=\tilde{\qR} \tilde{\qT}^{\PZF}  \tilde{\qR}^{H}.
 \end{align}
 As a result, \eqref{eq:mu1} can be tightly approximated as
\begin{align}\label{eq:lmu}
\mu^{i} &\approx 
|\alpha|^{2} L P_{R} \trac\left(\qA \qA^{H}\left(\zeta_i\trac(\tilde{\qT}^{i})\qI_N + \sigma_{R}^{2}\qI_N \right)^{-1}\right)\nonumber\\
&=\frac{|\alpha|^{2} L P_{R} \trac(\qA \qA^{H})}{ \zeta_i\trac\big(\tilde{\qT}^{i} \big)+\sigma_{R}^2},
\end{align}
where $\zeta_{\MR} = \zeta_{\ZF} =\beta_{br}$ and $\zeta_{\PZF} = (\beta_{br}-\gamma_r)$.
We now derive closed-form expressions of $\mu$  for the $\MR$, $\ZF$, and $\PZF$ precoding schemes in the following subsections.

\subsubsection{Maximum-Ratio Precoder:}
From \eqref{e:tmr}, 
 $ \tilde{\qT}^{\MR}$ can be defined as
\begin{align}
\tilde{\qT}^{\MR}&=  \sum_{k=1}^{K} \qt_{k}^{\MR} \eta_k(\qt^{\MR}_{k})^H = \qT^{\MR}\qD_{\eta} (\qT^{\MR})^H \nonumber\\
&=\frac{1}{M}( {\qH^{*}} \qD_{\eta}{\qH^{T}}).
\end{align}

Since $\frac{1}{M} {\qH^{T}} {\qH^{*}}\to \qI_K$ for sufficiently large $M$~\cite{Ngo:TCOM:2013}, we have $\trac\big( \tilde{\qT}^{\MR}\big) =\trac\left(\qD_{\eta} \frac{{\qH^{T}} {\qH^{*}}}{M} \right)\approx\sum_{k=1}^{K} \eta_{k}$. As a result,
\vspace{-0.5em}
\begin{equation}\label{P:muMR}
\mu^{\MR}\approx\frac{ |\alpha|^{2} L P_{R} \mathbf\trac\big(\qA \qA^{H}\big)} {\beta_{br}\sum_{k=1}^{K} \eta_{k} +\sigma_{R}^{2}}.
\end{equation}

\vspace{-1em}
\subsubsection{Zero-Forcing Precoder:}
Form \eqref{eq:tzf}, $\tilde{\qT}^{\ZF}$ can be derived as
\begin{align}~\label{eq:tildeTZF}
\tilde{\qT}^{\ZF}&=\qT^{\ZF}\qD_{\eta}(\qT^{\ZF})^H\nonumber\\
&=(M-K)\qH^{*}( {\qH}^{T} {\qH}^{*})^{-1}\qD_{\eta}
 ({\qH}^{T} {\qH}^{*})^{-1}\qH^{T}.
\end{align}
To this end, by using~\eqref{eq:tildeTZF}, we get
\begin{align}
\trac\big(\tilde{\qT}^{ZF}\big) &=(M-K)\trac\big( \qD_{\eta}({\qH}^{T} {\qH}^{*})^{-1}\big) 
\nonumber\\
&\stackrel{(a)}{\approx}\frac{M-K}{M}  \trac\left( \qD_{\eta}\right) \nonumber\\
&\stackrel{(b)}{=}\frac{M-K}{M} \sum_{k=1}^{K} \eta_{k},
\end{align}
where (a) follows from the fact that $({\qH}^{T} {\qH}^{*})^{-1}\approx\frac{1}{M}\qI_K$ for sufficiently large values of $M$, and (b) holds since $\trac\left(\qD_{\eta}\right)=\sum_{k=1}^{K} \eta_{k}$. Then, the non-centrality parameter $\mu^{\ZF}$  can be approximated by
\vspace{-0.5em}
\begin{align} \label{P:muZF}
\mu^{\ZF}\approx \frac{|\alpha|^{2} L P_{R} \trac\big(\qA \qA^{H}\big)}{  \frac{M-K}{M}\beta_{br}\sum_{k=1}^{K}\eta_{k} +\sigma_{R}^{2}}. 
 \end{align}
\vspace{-1.5em}
\subsubsection{Protective Zero-Forcing Precoder:}
\begin{proposition}~\label{prop:pd:PZF}
With PZF precoding, the non-centrality parameter,  $\mu$, can be derived as
\begin{align} \label{P:muPZF}
\mu^{\PZF}\approx\frac{ | \alpha |^{2} L P_{R} \mathbf\trac\big(\qA \qA^{H}\big)} { \frac{M-K}{M}(\beta_{br}-\gamma_r)\sum_{k=1}^{K} \eta_{k} +\sigma_{R}^{2}}.
\end{align}
\end{proposition}
\vspace{-2em}
\begin{proof}
    See Appendix~\ref{app:PZF:pd}.
\end{proof}

\vspace{0em}
 \section{Power Optimization}~\label{sec:power allocation}
 To protect the radar system against the BS transmissions and at the same time to guarantee the performance requirements of the cellular users, power control at the BS and radar can be applied. In this regard, we aim at selecting the radar transmit power $P_{R}$ and BS power control coefficients $\eta_{k}$ to maximize the  detection probability for MIMO radar, under the constraints on per-user SE and sum of power control coefficients $\eta_{k}$.
More precisely, the optimization problem is mathematically described as 
\begin{subequations}\label{P:SE}
	\begin{align}
	(\mathbb{P}1):\hspace{1em}
	\underset{{P_{R}, \eta_{k}}}{\max}\,\, \hspace{1em} &
		P_d
		\\
		\mathrm{s.t.} \,\, \hspace{1em}
		& {\SE}_k \geq \SEkth,~ \forall k, 
		\\
	& 0\leq {P}_{R}\leq\PRmax,
		\label{P1:c2}
		\\
		&\sum_{k=1}^{K}\eta_{k} \leq 1,
		\label{P1:c3}
	\end{align}
\end{subequations}
where $\PRmax$ is the maximum transmit power at the MIMO radar and $\SEkth$ is the minimum SE requirement by the $k$-th user.  The constraint \eqref{P1:c3} represents the maximum transmit power constraint at the BS. By invoking~\eqref{eq:Tq}, and since the Marcum $Q$-function is an increasing function of $\mu$, we can further rewrite $(\mathbb{P}1)$ as the following optimization problem
\vspace{-0.2em}
\begin{subequations}\label{P:SE}
	\begin{align}
	(\mathbb{P}2):\hspace{1em}
	\underset{{P_{R}, \eta_{k}}}{\max}\,\, \hspace{1em} &
		\mu
		\\
		\mathrm{s.t.} \,\, \hspace{1em}
		& {\SE}_k \geq \SEkth,~ \forall k, 
		\\
		& 0\leq {P}_{R}\leq\PRmax,
		\label{P:c2}
		\\
		&\sum_{k=1}^{K}\eta_{k} \leq 1.
		\label{P2:c3}
	\end{align}
\end{subequations}

By invoking~\eqref{P:muMR},~\eqref{P:muZF}, and~\eqref{P:muPZF}, we can rewrite the optimization problem $(\mathbb{P}2)$ in the general form
\vspace{0em}
\begin{subequations}\label{P:SE:general}
	\begin{align}
	(\mathbb{P}3):\hspace{1em}
	\underset{P_{R}, \eta_{k}}{\max}\,\, \hspace{1em} &
		\frac{ |\alpha|^{2} L P_{R} \trac\big(\qA \qA^{H}\big)} { a_i   \sum_{k=1}^{K} \eta_{k}+\sigma_{R}^{2}}
         \label{P3:c1}
		\\
		\mathrm{s.t.} \,\, \hspace{1em}
		& \frac{ b_i\rho \eta_{k}\gamma_{k}}{\rho c_{i,k} \sum_{k'=1}^K\eta_{k'}\!+\!\bar{\beta}_k P_{R} N\!+\! 1}\geq \gamkth,~ \forall k, 
		\\
		& 0\leq {P}_{R}\leq\PRmax,
		\label{P:c2}
		\\
		&\sum_{k=1}^{K}\eta_{k} \leq 1,
		\label{P3:c3}
	\end{align}
\end{subequations}
where $\gamkth=2^{\SEkth}-1$, $i\in\{\MR, \PZF, \ZF\}$ and $a_{\MR} = \beta_{br}$,  $a_{\ZF} =\frac{(M-K)}{M}\beta_{br}$, $a_{\PZF} = \frac{M-K}{M} (\beta_{br}-\gamma_r)$, $b_{\MR} = M$; $b_{\PZF} = \frac{1}{M}(M-K)(M-N)$, $b_{\ZF} = (M-K)$; $c_{\MR,k} =\gamma_{k}$, $c_{\PZF,k} = b_{\PZF}(\beta_{k} - \gamma_k)$, and $c_{\ZF,k} =\beta_{k} - \gamma_k$.

\subsection{Bisection Search-Based Solution}
Since expression  \eqref{P3:c1} is quasi-concave and constraints of the problem are linear functions of the optimization variables, the optimization problem ($\mathbb{P}3$) is quasi-concave. Thus, it can be equivalently reformulated as
\begin{subequations}\label{P:P4}
	\begin{align}
	(\mathbb{P}4):\hspace{0.5em}
	\underset{{P_{R}, \eta_{k}}}{\max}\,\, \hspace{1em} &
		t
		\\
		\mathrm{s.t.} \,\, \hspace{0.5em}
		&{ b_i\rho \eta_{k}\gamma_{k}}   \geq\nonumber\\
  &\hspace{0em}\gamkth   \Big({\rho c_{i,k} \sum_{k'=1}^K\eta_{k'}\!+\!\bar{\beta}_k P_{R} N\!+\! 1} \Big), \forall k,
		\\
&{\left |\alpha|^{2} L P_{R} \trac\right(\qA \qA^{H})}\! \geq\! t\Big( { a_i   \sum_{k=1}^{K} \eta_{k}+\sigma_{R}^{2}}\Big)
            \\
		& 0\leq {P}_{R}\leq\PRmax,
		\label{P:c2}
		\\
		&\sum_{k=1}^{K}\eta_{k} \leq 1,
		\label{P4:c3}
	\end{align}
\end{subequations}
where 
\begin{align}
t=\frac{\left |\alpha|^{2} L P_{R} \mathbf\trac\right(\qA \qA^{H})} {a_i  \sum_{k=1}^{K} \eta_{k}+\sigma_{R}^{2}}, 
  \end{align}
is an auxiliary variable.

Problems~\eqref{P:P4}  can be solved efficiently by a bisection search\cite{boyd2004convex}, where in each step we solve a sequence of convex feasibility problems  as detailed in \textbf{Algorithm 1}.

\begin{algorithm}[!t]
\caption{Bisection Algorithm for Solving $(\mathbb{P}4) $}
\begin{algorithmic}[1]
\label{alg:GR}
\STATE 
Initialization: choose the initial values of $t_{\min}$  and $t_{\max}$, where $t_{\min}$  and $t_{\max}$  define a range of relevant values of the objective function, choose a tolerance $ \varsigma \geq 0$.
 \STATE 
 While ${t_{\min}-t_{\max}}\leq \varsigma $.
  \STATE  Set $t =\frac{t_{\min}+t_{\max}}{2}$ and solve the feasibility problem.
  \STATE   If problem $(\mathbb{P}4)$ is feasible, 
  then set $t_{\min}=t$, else set $t_{max}=t$.
\STATE Stop if $t_{\max} - t_{\min}\leq \epsilon$. Otherwise, go to Step 2.
\end{algorithmic}
\end{algorithm}

\subsubsection{Convergence and Complexity Analysis}
Let $p^\star$ denote the optimal value of the quasiconvex optimization problem ($\mathbb{P}3$). If the feasibility problem ($\mathbb{P}4$) is feasible, then we have $p^\star \leq t$. Conversely, if the problem ($\mathbb{P}4$) is infeasible, then we can conclude that $p^\star \geq t$. We can assess whether the optimal value, denoted as $p^\star$, for a quasiconvex optimization problem is less than or greater than a predefined value $t$ by solving the convex feasibility problem ($\mathbb{P}4$). If the convex feasibility problem yields a feasible solution, we can conclude that $p^\star \leq t$. Conversely, if the convex feasibility problem is infeasible, it implies that $p^\star \geq t$. To implement a bisection algorithm, the interval $[t_{\min},t_{\max}]$ is guaranteed to contain $p^\star$, i.e., we have $t_{\min} \leq p^\star \leq t_{\min}$ at each step. In each iteration, the interval is divided in two, i.e., bisected, so the length of the interval after $k$ iterations is $\frac{(t_{\max}-t_{\min})}{2^k}$, where $(t_{\max}-t_{\min})$ is the length of the initial interval. It follows that exactly $\lceil\log2((t_{\max}-t_{\min})/\epsilon)\rceil$ iterations are required before the algorithm terminates. Each step involves solving the convex feasibility problem ($\mathbb{P}4$).

According to~\cite{Tam2016TWC}, the  per-iteration cost to solve the feasibility problem ($\mathbb{P}4$) is $\mathcal{O}\big( (n_{l} + n_{v}) n_v^2n_{l}^{0.5}\big)$, where $n_{l} =K+3 $ denotes the number of linear constraints and $n_v =K+1$ is the number of real valued scalar decision variables. Therefore, the overall complexity of the bisection algorithm is $\lceil\log2((t_{\max}-t_{\min})/\epsilon)\rceil\mathcal{O}\big( (n_{l} + n_{v}) n_v^2n_{l}^{0.5}\big)$.

\subsection{ Linear Programming Solution}
The computational complexity of the bisection-based search algorithm can be reduced by noticing that ($\mathbb{P}$1) can be presented as a linear programming (LP) problem.  The optimization problem in~\eqref{P:SE:general} is a non-concave problem due to non-concavity of the objective function. Before proceeding, we notice that both the objective function and first constraint~\eqref{P3:c1} are coupled together via the optimization variables $P_R$ and $\eta_k$, for $k=1,\ldots,K$. We now prove that the optimal value for $P_{R}$ is  $P_{R}^{\star}=\PRmax$. To this end, assume that $P_{R}^{\star}\geq\PRmax$. Now, we define the  new variables $P_{R,\new}=cP_{R}^{\star}$ and $\eta_{k,\new}=c\eta^{\star}_k$, where $c \geq 1$. Then, the objective function~\eqref{P3:c1} can be obtained as
\begin{align}
   f(P_{R,\new},\{\eta_{k,\new}\}) &=\frac{ |\alpha|^{2} L cP^*_R \trac\big(\qA \qA^{H}\big)} { a_i   c\sum_{k=1}^{K} \eta^*_k+\sigma_{R}^{2}}
   \nonumber\\
   &\geq \frac{ |\alpha|^{2} L cP^\star_R \trac\big(\qA \qA^{H}\big)} { a_i   c\sum_{k=1}^{K} \eta^*_k+\sigma_{R}^{2}}  = f(P_{R}^{\star},\{\eta_{k}^{\star}\}).
\end{align}
Moreover, at given $P_{R,\new}$ and $\eta_{k,\new}$, the first constraint is upper bounded as 
\begin{align}
&\frac{ b_i\rho c\eta^\star_k\gamma_{k}}
{\rho c_{i,k} c\sum_{k'=1}^K\!\!\eta^\star_{k'} \!+\!\bar{\beta}_k cP^*_R N \!+\!1}
\geq 
\nonumber\\
&\hspace{4em}
\frac{ b_i\rho \eta^\star_k\gamma_{k}}
{\rho c_{i,k} \sum_{k'=1}^K\eta^\star_{k'} \!+\!\bar{\beta}_k P^*_R N +1},
~ \forall k.
\end{align}
Therefore, by increasing $P_R$ and $\eta_k$, respectively, beyond $\PRmax$ and $\eta_k^\star$, the objective function is increased, and the first constraint is still satisfied. This implies that at the optimal point we have $P_{R}^{\star}=\PRmax$. 

Now, for the sake of notational simplicity, let $\boldsymbol{1}_K=[1,\ldots,1]^T$ be a vector of size $K\times 1$. Moreover, define $\BLETA = [\eta_1,\ldots,\eta_K]^T$. Then, the optimization problem ($\mathbb{P}$3) can be reduced to
\begin{subequations}\label{P5:SE:general2}
	\begin{align}
	(\mathbb{P}5):\hspace{0em}
	\underset{ \BLETA}{\max}\,\, \hspace{0.2em} &
    \boldsymbol{1}_K^T \boldsymbol{\eta}   +\frac{\sigma_{R}^{2}}{ a_i}\label{P5:obj}        
		\\
		\mathrm{s.t.} \,\, \hspace{0.2em}
		& \qC\BLETA \leq \qb ,\label{P5:c1}
	\end{align}
\end{subequations}
where $\qb = [\Lambda_1,\ldots,\Lambda_K,1]$, with $\Lambda_k= -\gamkth(\bar{\beta}_k \PRmax N\!+\! 1)$ and
\begin{align}\label{P:Cmatrix}
    \qC =
    \begin{bmatrix}
\Psi_{1} & \Theta_{1} & \Theta_{1} & \ldots
& \Theta_{1} & \Theta_{1} \\
\Theta_{2} & \Psi_{2}& \Theta_{2}  & \ldots
& \Theta_{2} & \Theta_{2} \\
\vdots & \vdots & \vdots & \ddots
& \vdots & \vdots \\
\Theta_{K} & \Theta_{K} & \Theta_{K} & \ldots
& \Theta_{K} & \Psi_{K}
\\
1 &1 &1 &\ldots &1 &1
\end{bmatrix},
\end{align}
with $\Theta_{k}=\rho c_{i,k} \gamkth$ and $\Psi_{k}=\rho c_{i,k} \gamkth-\rho b_i \gamma_{k}$. The problem ($\mathbb{P}$5) is an LP problem and can be efficiently solved via interior-point algorithms~\cite{boyd2004convex}. We use the \textit{linprog} function in Matlab platform to solve Problem ($\mathbb{P}5$).

According to~\cite{Tam2016TWC}, the complexity of  the optimization problem ($\mathbb{P}5$) is $\mathcal{O}\big( (n_{l} + n_{v}) n_v^2n_{l}^{0.5}\big)$, where $n_{l} =K+1 $ denotes the number of linear constraints while $n_v =K$ is the number of real valued scalar decision variables.

\section{Extensions}
\subsection{Multi-target Scenarios}
The considered spectrum sharing MIMO radar and communication system has interesting application scenarios. For instance, identifying abnormal or suspicious activities in a specific area and sending an alert is one potential scenario. Therefore, the coexistence of co-located MIMO radar for single-target detection and communication has been the subject of several research studies~\cite{liu2017robust,8962251}. On the other hand, multi-target detection can be achieved through the multibeam/beampattern design at the MIMO radar~\cite{Liu:CSTO:2022}. Accordingly, new performance metrics such as beampattern gain are used in the literature to study the performance of such designs.

The primary concept of beampattern design is to determine the covariance matrix of the radar probing signals through convex optimization problems. Let $\qR_w = \frac{1}{L} \sum_{l=1}^L \qs_l\qs_l^H$ denote the covariance matrix of the probing signals, where $\qs_l$ is the $l$-th snapshot across the radar antennas. Three different optimization problems have been proposed in the literature to determine $\qR_w$. More specifically, the authors in~\cite{liu2018mu} formulated a constrained least-squares problem to approach an ideal beampattern as
\begin{subequations}~\label{eq:beampattern}
   \begin{align}
&\min \limits _{\alpha,{\qR_w}} ~\sum \limits _{m = 1}^{M} {{{\left |{ {\alpha {\tilde P_{d}}\left ({{\theta _{m}} }\right) - {{\mathbf {a}}^{H}}\left ({{\theta _{m}} }\right){\qR_w\qa}\left ({{\theta _{m}} }\right)} }\right |}^{2}}} \\&~s.t.~\mathrm {diag} \left ({{\qR_w} }\right) = \frac {{P_{0}{\mathbf {1}_N}}}{N}, \\&\qquad ~{\qR_w} \succeq \mathbf{0},{\qR_w} = {{\qR_w}^{H}}, \\&\qquad ~\alpha \ge 0\\
&\qquad\trac(\qf_i^*\qf_i^T\qR_w) =0,~\forall i,~\label{inteptrotect}
\end{align} 
\end{subequations}
where $\{\theta_m\}_{m=1}^{M}$ is defined as a fine angular grid that covers the detection angle range of $[-\frac{\pi}{2},\frac{\pi}{2}]$, $\qa(\theta_m)=[1,e^{-j 2\pi d \sin(\theta_m)}, \ldots, e^{-j 2\pi d(N-1)  \sin(\theta_m)}]^T \in \mathbb{C}^{N \times 1}$ is the steering vector of the transmit antenna array, ${\tilde P_{d}}\left ({{\theta _{m}} }\right)$ is the desired ideal beampattern gain at $\theta_m$, $P_{0}$ is the power budget, $\alpha$ is a scaling factor, and $\boldsymbol{1}_N\in\mathbb{R}^{N\times 1}$. We notice that via constraint~\eqref{inteptrotect},  we force the radar signals to fall into the nullspace of the channel between the radar antennas and downlink users.

Nevertheless, finding the detection probability of the radar is challenging. We leave the performance evaluation of the multi-target scenario for future investigation. 
\subsection{Correlated Fading}
In this subsection, we extend our analysis to consider the spatially correlated Raleigh fading scenario. The correlated Rayleigh channel vector between the $k$-th user and the BS can be modeled as
\begin{equation} 
{{ {\mathbf g}}_{k}} \sim {\mathcal {CN}}\left ( \boldsymbol{0},\boldsymbol{\Omega}_k\right),
\end{equation}
where $\boldsymbol{\Omega}_k\in\mathbb{C}^{M\times M}$  is the transmit
correlation matrix specific to terminal $k$.

\textbf{Training phase:}  Based on the observable pilot matrix in~\eqref{eq:Pilotmat}, the BS correlates ${\mathbf {Y}}_{\mathrm {P}}$ with the pilot sequence $\pmb{\varphi}_{p,k}$ of user $k$, leading to the processed pilot sequence
\begin{align}~\label{eq:ykp}
\ykp =  \qY_{\mathrm{\text p}}{{\pmb{\varphi}}_{p,k}} = \sqrt {\tau_{\mathrm {p}} {\rho _{u}}}   {{{ {\mathbf g}}_{k}}} + {{ {\mathbf N}}}{{\pmb{\varphi}}_{p,k}},\end{align}
where $\pmb{\varphi}_{p,k}\in \mathbb{C}^{\tau _{\mathrm {p}} \times 1}$ indicates the pilot signal assigned to the $k$-th user. Based on the observation $\ykp$, the MMSE estimate of the channel $\qg_k$ is given by
\begin{equation*} {{ {\hat {\qg}}}_{k}} = \sqrt {\tau_{\mathrm {p}} {\rho _{u}}} \boldsymbol{\Omega}_{k}{{\boldsymbol{\Xi }}_{k}}\ykp,\end{equation*}
where
\begin{align} \label{eq:Ek}
{{\boldsymbol{\Xi }}_{k}} = {\bigg ({{\tau_{\mathrm {p}} \rho_u \boldsymbol{\Omega}_{k} + {{ {\mathbf I}}_{M}}} }\bigg)^{ - 1}}.
\end{align}
Therefore, the MMSE estimated channel ${{ {\hat {\mathbf g}}}_{k}}$ and the corresponding estimation error ${{ {\tilde {\mathbf g}}}_{k}}$ are complex Gaussian random vectors, distributed as follows
\begin{align*} {{ {\hat {\mathbf {g}}}}_{k}}\sim&{\mathcal {N}_{\mathbb {C}}}\left ({{{\pmb {0}},\;{{ {\mathbf Q}}_{k}}} }\right), \\ {{ {\tilde {\mathbf {g}}}}_{k}}\sim&{\mathcal {N}_{\mathbb {C}}}\left ({{{\pmb {0}},{{ {\mathbf C}}_{k}}} }\right),\end{align*}
where $\qQ_k=\boldsymbol{\Omega}_k-\qC_k$ and $\qC_k$ is expressed as
\begin{align}~\label{eq:Ck} 
{{ {\mathbf C}}_{k}} =\boldsymbol{\Omega}_{k} - \tau_{\mathrm {p}} \rho_u\boldsymbol{\Omega}_{k}{{\boldsymbol{\Xi }}_{k}}\boldsymbol{\Omega}_{k}.
\end{align}

The following Proposition provides the SE at the $k$th cellular user and detection probability at the radar over the spatially correlated fading environments and with MR precoding design at the BS.

We proceed to find closed-form expression for $\SINR_k$ with $\MR$ precoding. 
\begin{proposition}~\label{propos:MR:correlated}
For the considered spatially correlated Rayleigh
channels, the SE of the $k$-th user achieved by the $\MR$ precoding is given by~\eqref{eq:SINE:MR:correlatedPC} at the top of the next page.
 \begin{figure*}
\begin{align}~\label{eq:SINE:MR:correlatedPC}
\SEk^{\MR,\cor}=\Big(1-\frac{ \tau_{\mathrm {p}}}{\tau}\Big)\log _{2}\left(1+\frac{
                {{\eta_k}\rho } \tau_{\mathrm {p}} \rho_u \trac(\boldsymbol{\Omega}_{k}\boldsymbol{\Xi }_{k}\boldsymbol{\Omega}_{k})
                 }
                 {
                 \sum_{k'=1}^K
                 \eta_{k'} \rho                 
                 \frac{\trac( \boldsymbol{\Omega}_{k'}{{\boldsymbol{\Xi }}_{k'}}\boldsymbol{\Omega}_{k'}\boldsymbol{\Omega}_k)}
                 {\trac( \boldsymbol{\Omega}_{k'}\boldsymbol{\Xi }_{k'} \boldsymbol{\Omega}_{k'})}
 +  P_R\bar{\beta}_k N +  \sigma_{C}^2}\right),
\end{align}
	\hrulefill
	\vspace{-1mm}
\end{figure*}


Moreover, the detection probability at the radar is given by~\eqref{eq:detection}, where
\begin{align}\label{eq:lmu2:MRcor}
\mu^{\MR,\cor} &\approx 
\frac{|\alpha|^{2} L P_{R} \trac(\qA \qA^{H})}{ \zeta_{\MR}\sum_{k=1}^{K}
 \eta_k+\sigma_{R}^2}.
\end{align}
\end{proposition}
\begin{proof}
    See Appendix~\ref{proof:propos:MR:correlated}.
\end{proof}

We notice that, in the case of uncorrelated fading channels $\trac(\boldsymbol{\Omega}_k{{\boldsymbol{\Xi }}_{k}}\boldsymbol{\Omega}_{k}) = \frac{M\gamma_k}{\tau_{\mathrm {p}} \rho_u}$ and 
$\trac( \boldsymbol{\Omega}_{k'}{{\boldsymbol{\Xi }}_{k'}}\boldsymbol{\Omega}_{k'}\boldsymbol{\Omega}_k)= \frac{M\gamma_k\gamma_{k'}}{\tau_{\mathrm {p}} \rho_u}$. Therefore, $\SEk^{\MR,\cor}$ reduces to $\SEk^{\MR}$ in~\eqref{Prop:SINR:MR}. 
Moreover, by comparing~\eqref{eq:lmu2:MRcor} and~\eqref{P:muMR}, we observe that when using MR precoding at the BS, the same detection probability is achieved over both correlated and uncorrelated Rayleigh fading channels.

\begin{Remark}
   Finding closed-form expressions for the achievable SE and detection probability  for  $\ZF$ and $\PZF$ precoding designs seems intractable. However, the corresponding results can be obtained through numerical simulations.   
\end{Remark}

\vspace{2em}
\section{Numerical results}~\label{sec:numerical}
In this section, we provide numerical results to verify our analysis and evaluate the performance of the proposed power control design. We assume that the MIMO radar system deploys a uniform linear array  with inter-antenna spacing $d = \lambda/2$. The radar SNR is defined as $\mathtt{SNR}_R=\frac{L |\alpha|^{2}  P_{R} }{\sigma_{R}^{2}}$ \cite{liu2017robust}. The target is set to be located at the direction of $ \theta=10^{\circ}$. In the communication system, we consider large-scale fading with path loss, shadow fading, and random user locations. To generate the large-scale fading coefficients, we adopt the model proposed in \cite{ngo2017no}. For the sake of understanding, let us visualize a circular cell with a radius of approximately $15$ kilometres. At the center of the cell, the BS equipped with $M$ antennas serves $K$ users randomly distributed within the cell. Furthermore, a MIMO radar is positioned $10$ kilometres away from the BS. We model the large-scale fading coefficient $\beta$ as~\cite{ngo2017no}
\begin{equation} 
\beta 
= 
\text {PL}_{0}\left ({\frac {d}{R_{\text {min}}}}\right )^\upsilon \times 10^{\frac {\sigma_{\text {sh}} z }{10}}, 
\end{equation}
where $\beta\in\{\beta_k, \bar{\beta}_k, \beta_{br}\}$,  $d\in\{d_{k}, \bar{d}_{k}, d_{br}\}$ with $d_k$ ($\bar{d}_{k}$) being the distance between the BS (radar) and user $k$, while $d_{br}$ is the distance between the BS and radar;  ${PL}_{0}$ denotes a reference path loss constant which is chosen to satisfy a given downlink cell-edge SNR;  $\upsilon$ and ${\sigma _{\text {sh}}}$ represent the path loss exponent and  standard deviation of the shadow fading, respectively; and $z\sim\mathcal{N}(0,1)$. In our examples, we chose 
$\upsilon$ = 3.8, ${\sigma _{\text {sh}}}  = 8$ dB. 

The $(i,j)$ entry in the correlation matrix of terminal $k$ is given by~\cite{Tataria:WCL:2017}
\begin{align} \left [{ \boldsymbol {\Omega}_{k}}\right ]_{i,j}=\frac {1}{2\Delta } \int \nolimits _{-\Delta +\phi _{k}}^{\Delta +\phi _{k}}e^{-j2\pi {}d\left ({i-j}\right ) \sin \left ({\theta _{k}}\right )}d\theta _{k}, \end{align}
where $\Delta$ denotes the azimuth angular spread, $\phi _{k}$ is the central azimuth angle from  BS array to the terminal $k$, $\theta _{k}$ is the actual angle-of-departure  and  $d\left ({i-j}\right )$ captures the inter-element spacing normalized by the carrier wavelength between $i$-th and $j$-th antenna elements. Unless explicitly stated, we set $d(1)=0.5$ and assume that $\phi _{k}$ drawn from a uniform distribution on $0$, $2\pi$, i.e., $\phi _{k}\sim \mathcal{U}[0, 2\pi]$. The instantaneous value of $\theta _{k}$ is also drawn from a uniform distribution on $-\frac{\Delta}{2}$, $\frac{\Delta}{2}$
, i.e., $\theta _{k}\sim \mathcal{U}[-\frac{\Delta}{2}, \frac{\Delta}{2}]$. As such, $\Delta$ represents the tota angular spread, naturally bounded from $0$ to $2\pi$ radians.

\begin{figure}  
	\includegraphics[width=0.52\textwidth]{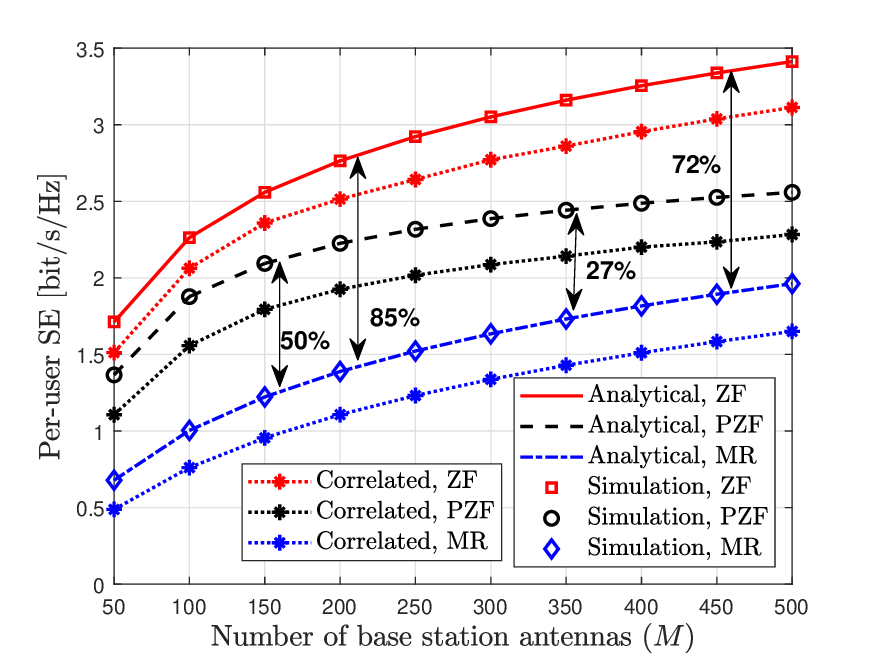}
	\centering
	\caption{Per-user SE versus the number of BS antennas, $M$, ($N =20$, $\rho =20$ dB, $K =20$, $\Delta=20^{o}$).}
	\label{fig:figure2}
\end{figure}

\begin{figure}
	\includegraphics[width=0.52\textwidth]{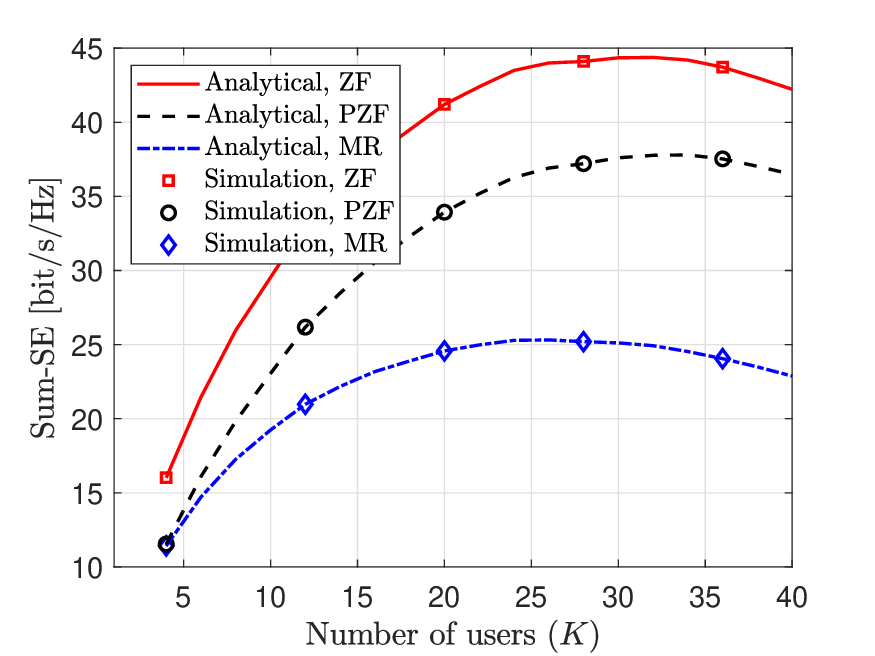}
	\centering
	\caption{Sum-SE versus the number of users, $K$, ($M = 200$, $\rho = 20$ dB, $N = 20$).}
	\label{fig:figure3}
\end{figure}

\begin{figure}  
	\includegraphics[width=0.52\textwidth]{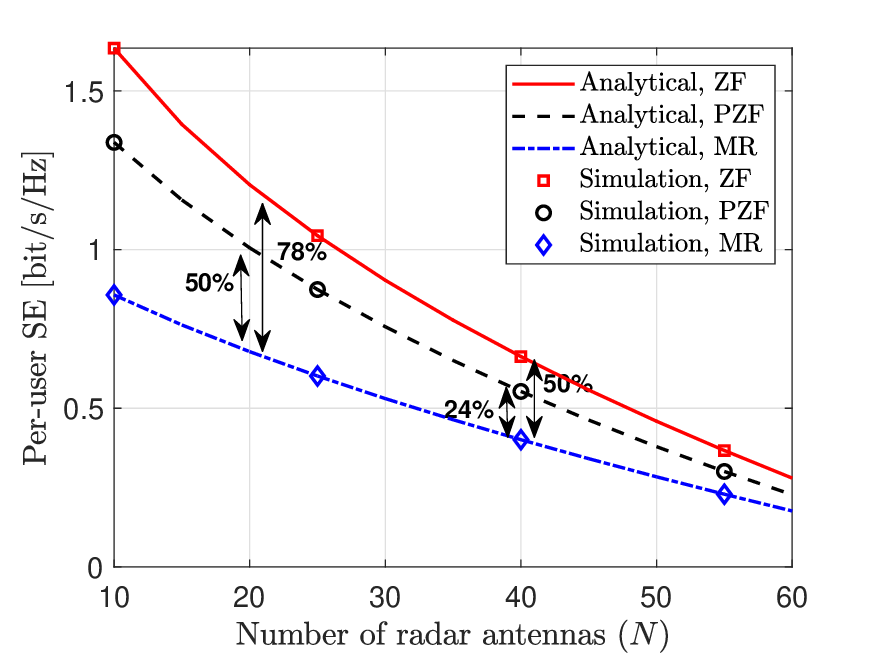}
	\centering
	\caption{Per-user SE versus the number of radar antennas, $N$, ($M =200$, $\rho =20$ dB, $K =20$).}
	\label{fig:figure4}
 \vspace{0em}
\end{figure}


Figure~\ref{fig:figure2} shows the downlink SE of different massive  MIMO systems versus the number of BS antennas $M$, for different precoding schemes $\MR, \ZF$ and $ \PZF$.  Our analysis has been confirmed by the simulation results, demonstrating a precise alignment between the analytical and simulated outcomes. As anticipated, an increase in the number of BS antennas directly corresponds to SE enhancement. We observe that the performance gaps between the considered precoding schemes are decreased, when the number of BS antennas increases. For example, by increasing the number of BS antennas from $M=200$ to $M=400$, the gap between the $\ZF$ ($\PZF$) scheme and $\MR$ is reduced from $85\%$ ($50\%$) to $72\%$ ($27\%$).  This behavior can be explained as follows: By comparing~\eqref{Prop:SINR:ZF} and~\eqref{Prop:SINR:PZF}, we observe that by increasing the number of BS antennas $M$, for a fixed number of receive antennas at the radar, $N$, we have $\frac{M-N}{N} \xrightarrow{M \rightarrow \infty} 1$. Therefore, the numerators of the SINRs inside the logarithmic functions are approximately the same for the $\ZF$ and $\PZF$ schemes. However, the first term in the denominator of $\SINR^\PZF_k$ scales as $(M-K)$, whereas the corresponding term in $\SINR^\ZF_k$ scales as $1$. Therefore, by increasing $M$, the reduction in $\SINR^\PZF_k$ becomes greater compared to $\SINR^\ZF_k$. As a result, the gap between the SE achieved by $\ZF$ and $\PZF$ increases with higher values of $M$.
 Finally, when comparing the results between the correlated and uncorrelated channels, we find that the same insights apply to correlated environments.


\begin{figure}  
	\includegraphics[width=0.50\textwidth]{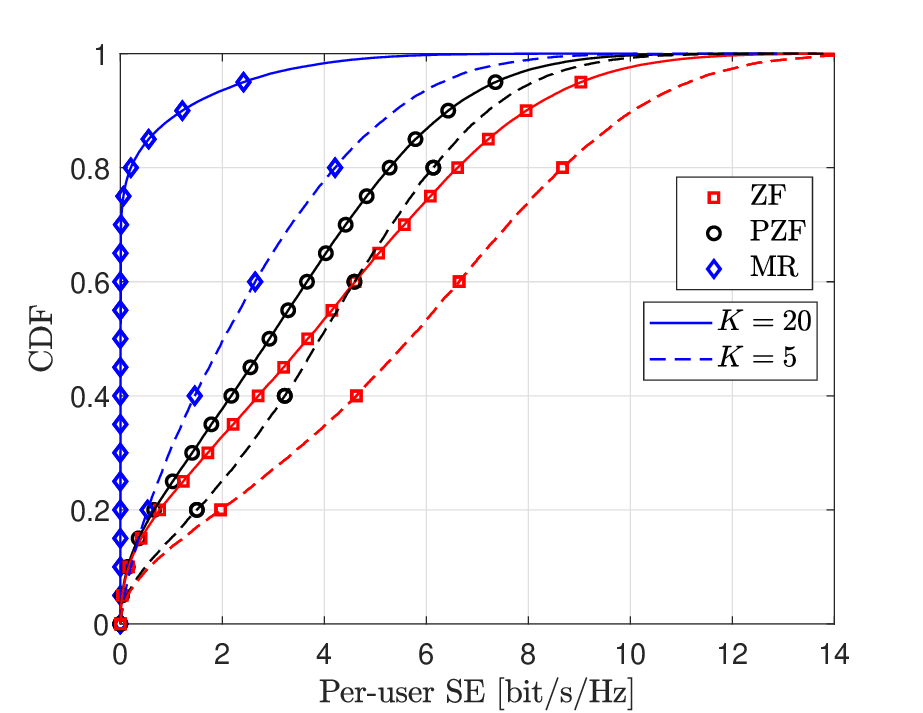}
	\centering
	\caption{CDF of the per-user SE for different number of users ($N =20$, $M=200$, $\rho =20$ dB).}
	\label{fig:figure5}
  \vspace{-1em}
\end{figure}

\begin{figure}
	\includegraphics[width=0.51\textwidth]{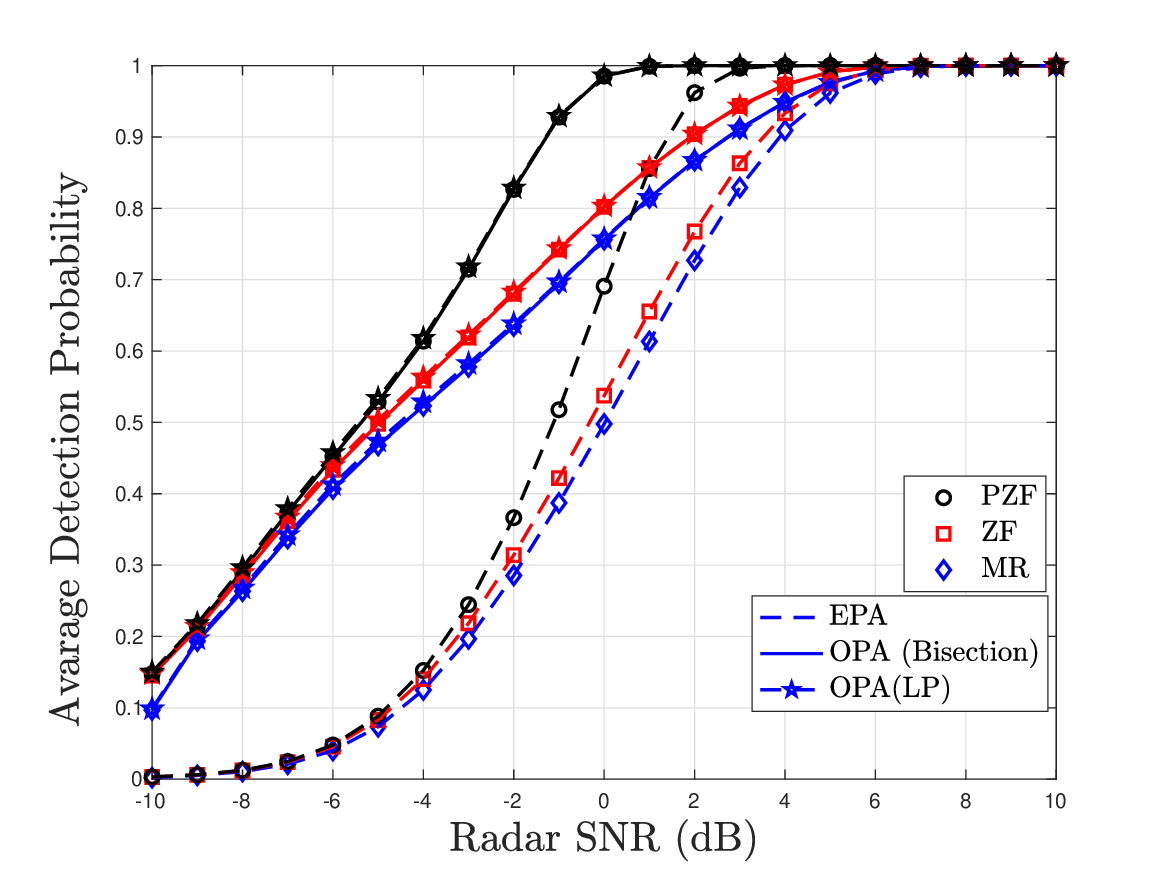}
	\centering
	\caption{Average detection probability versus the radar SNR ($M = 200$, $K= 20$, $N = 20$, $\mathtt{SNR}_R=0$ dB).}
	\label{fig:figure6}
\end{figure}
Figure~\ref{fig:figure3} shows the downlink sum-SE of massive  MIMO system versus the number of users $K$, for different precoding schemes, i.,e $\MR, \ZF$ and $ \PZF$. By increasing the number of users, sum-SE increases and tends to the maximum value and then starts to reduce due to increasing pilot overhead. Moreover, we observe that by increasing $K$, the gap between the $\ZF$ and $\MR$ increases, while the gap between the $\ZF$ and $\PZF$ remains constant, which is due to the severe increase of the inter-user interference in the case of $\MR$ processing. These results reveal the necessity of developing pilot reuse among different users and designing pilot assignment algorithms to manage the pilot contamination effects. 

To demonstrate the impact of radar interference on massive MIMO systems, we examine the SE as a function of the number of transmit antennas at the radar in Figure~\ref{fig:figure4}. 
As the number of antennas at the radar increases, the interference towards the massive MIMO system  increases, leading to  SE reduction. 
For different number of antennas at the radar, $\MR$ yields the worst performance, while $\ZF$ outperforms all schemes. Notably, $\ZF$ exhibits a remarkable improvement of 78$\%$  over $\MR$, while $\PZF$ demonstrates a substantial enhancement of 50$\%$  compared to $\MR$, when the number of radar antennas is $N=20$. When the number of radar antennas is increased to $N=80$, $\ZF$ achieves a significant improvement of up to $50\%$  over $\MR$, while $\PZF$ exhibits a substantial enhancement of $24\%$  over $\MR$. This further solidifies the superiority of $\ZF$ and $\PZF$ techniques over the $\MR$ scheme in the given context.

Figure~\ref{fig:figure5} shows the cumulative distribution of the per-user SE for different values of $K=5$ and $K=20$, with $M = 200$ and $\rho = 20$ dB. We can observe that the per-user SE in the case of $K = 5$ is higher than that in the case of $K = 20$. This difference can be attributed to the fact that when there are more users, inter-user interference becomes more prevalent. Additionally, it is evident from the figure that $\MR$ suffers more from inter-user interference compared to $\ZF$ and $\PZF$.

\begin{figure}
	\includegraphics[width=0.50\textwidth]{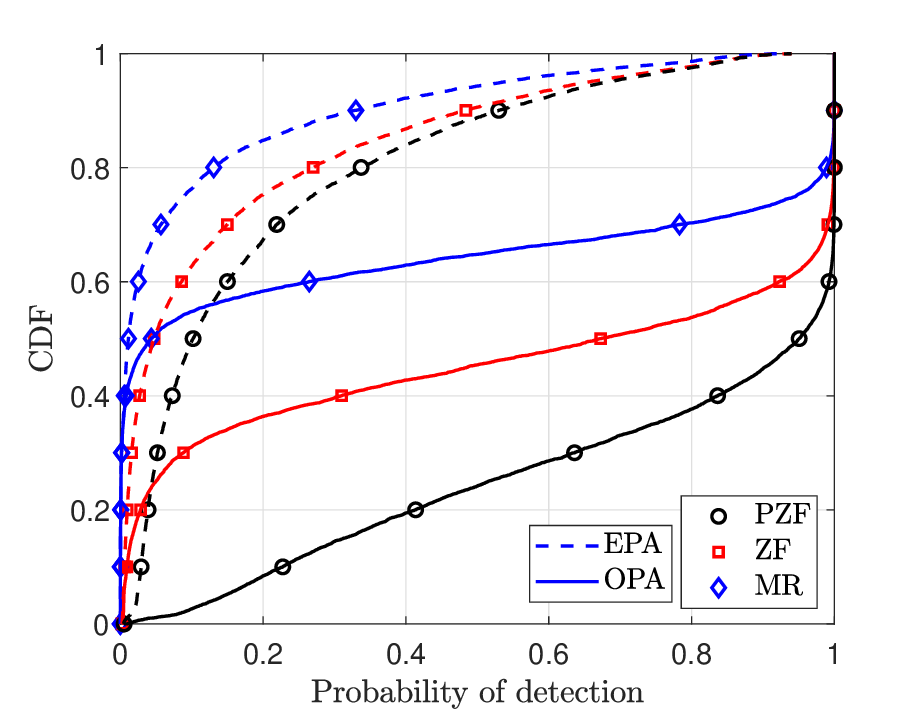}
	\centering
	\caption{CDF of detection probability ($M = 200$, $K= 20$, $N = 20$).}
	\label{fig:figure7}
 \vspace{-1em}
\end{figure}

\begin{figure}
	\includegraphics[width=0.50\textwidth]{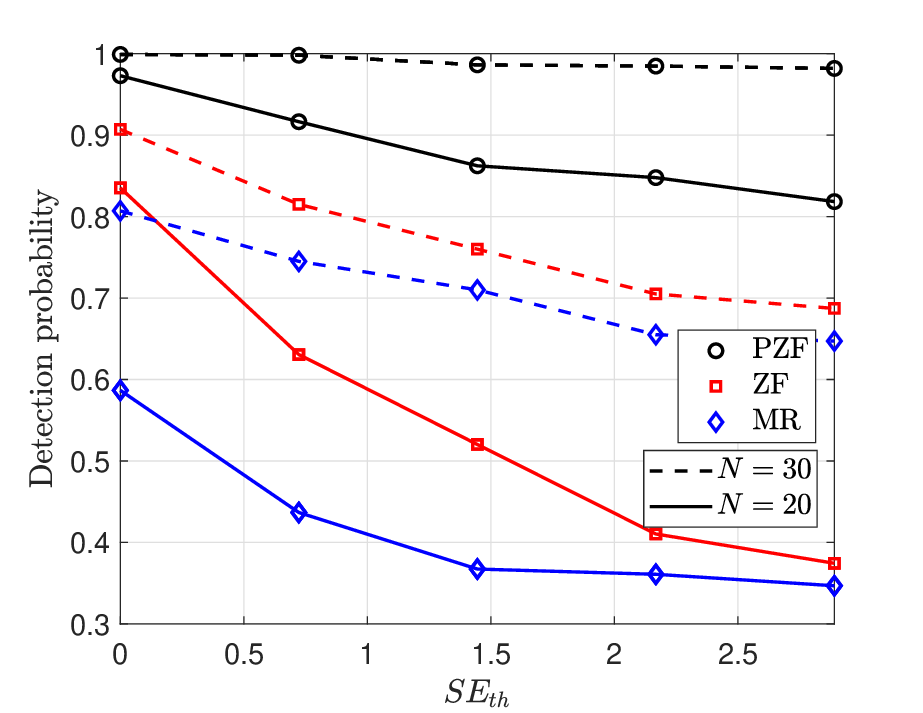}
	\centering
	\caption{Detection probability versus the SE threshold ($M = 200$, $K= 20$, $\mathtt{SNR}_R=0$ dB).}
	\label{fig:figure8}
\end{figure}

\begin{figure}
	\includegraphics[width=0.50\textwidth]{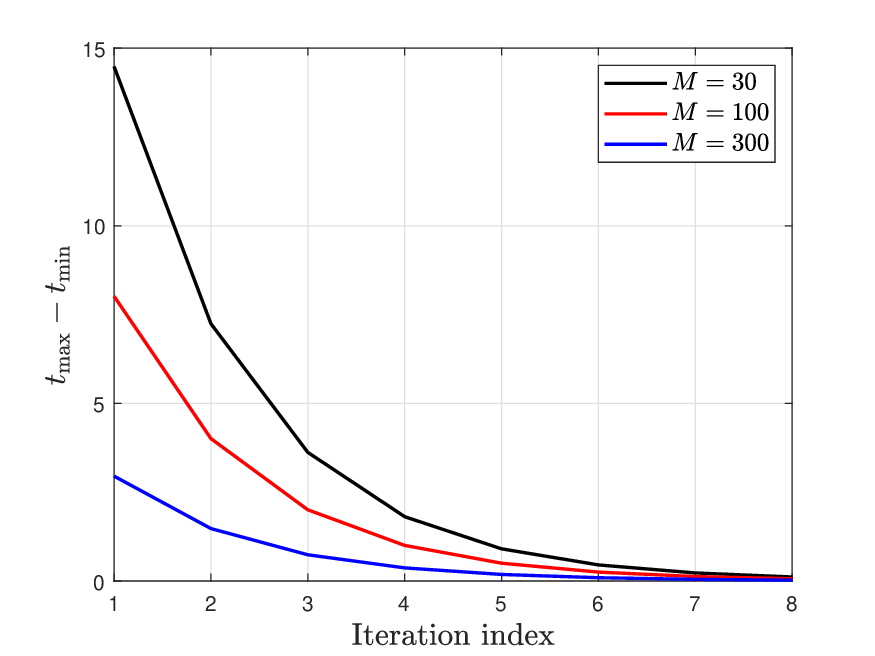}
	\centering
	\caption{ Convergence behavior of the bisection algorithm to solve the problem $(\mathbb{P}4)$ for different number of BS antennas ($K= 20$, $N = 20$).}
	\label{fig:figure9}
\vspace{-2em}
\end{figure}
\begin{figure}
	\includegraphics[width=0.50\textwidth]{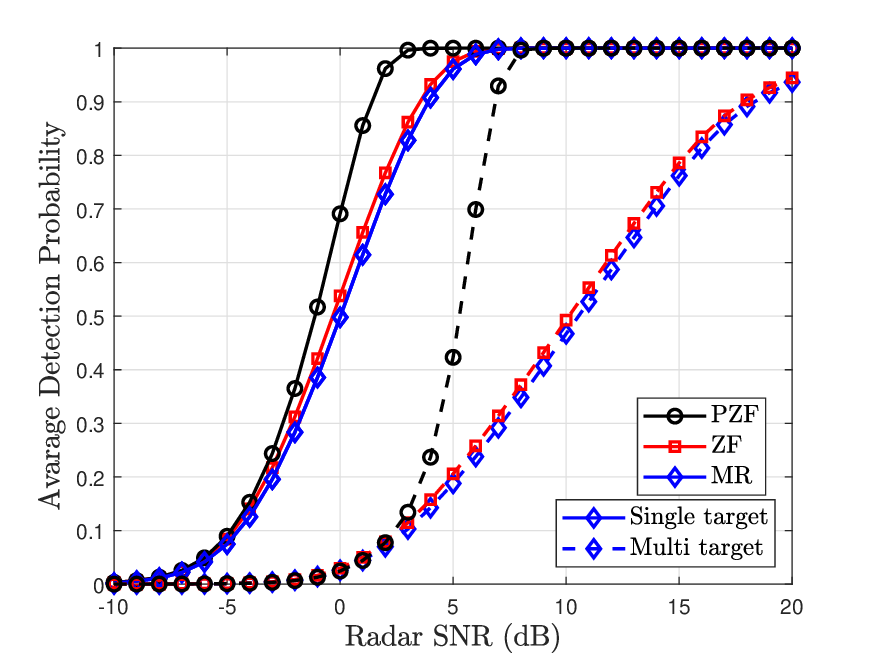}
	\centering
	\caption{Average detection probability versus the radar SNR for single and multiple target scenarios ($M = 200$, $K= 20$, $N = 20$, $\mathtt{SNR}_R=0$ dB)}
	\label{fig:figure10}
\vspace{-2em}
\end{figure}

Figure~\ref{fig:figure6} provides a comparative analysis of detection probabilities under two power allocation schemes: our proposed optimal power allocation (OPA) solutions and equal power allocation (EPA). It can be observed that the results obtained using the OPA via the bisection algorithm closely match those obtained using LP. In the EPA scenario, the radar transmits with maximum power level, i.e, $P_{R}=\PRmax$ and power control coefficients at the BS are set as $\eta_k=\frac{1}{K}$, $\forall k=1,\ldots,K$. To ensure a fair comparison, the selection of the SE threshold ($\SEth$) in the power allocation scheme is based on the EPA scenario. Among the considered schemes, our proposed $\PZF$ scheme yields the highest performance, followed by $\ZF$ and $\MR$, respectively. Moreover, our proposed power allocation scheme yields a remarkable enhancement in the  detection probability. For example, when the radar SNR is $2$ dB and OPA is applied, the detection probability  of $\PZF$, $\ZF$ and $\MR$ precoding schemes improved from $0.5$, $0.35$, and $0.3$ to $0.97$, $0.72$, and $0.68$, respectively.

Figure~\ref{fig:figure7} shows the CDF of the detection probability for $\PZF$, $\ZF$, and $\MR$ precoding designs and for OPA and EPA. It is evident that our proposed OPA strategy yields a substantial improvement in the detection probability performance and $\PZF$ design provides the best performance. 

Figure~\ref{fig:figure8} examines the trade-off between the detection probability and SE for $\PZF$, $\ZF$, and $\MR$ precoding designs and for two different radar antenna number. For a specific network realization, we first derive the maximum value for $\SEkth$, termed as $\SEkth^{\max}$, according to the EPA. Then, the optimization problem ($\mathbb{P}4$) is solved for different values of $\SEkth =[0, \SEkth^{\max}]$, and corresponding detection probability is derived. From Fig.~\ref{fig:figure8}, we can observe that for small values of $\SEkth$, the radar can transmit with more power, yielding better detection probability performance. Moreover, by increasing the number of radar receive antennas the detection probability is improved, while it degrades the SE of the users in the communication system (cf Fig.~\ref{fig:figure4}).

Figure~\ref{fig:figure9} shows value of $(t_{\max} - t_{\min})$ versus the number of iterations for one randomly generated set of channel realizations for different number of BS antennas. It can be observed that the algorithm with different number of BS antennas has a similar convergence speed and, as a matter of fact, converges to the optimal value after only a small number of iterations.

 Figure~\ref{fig:figure10} presents a comparison of average detection probabilities versus the radar SNR for both single and multiple target scenarios.  In the multiple target scenario, three targets are located at specific angular directions. The comparison is conducted for three precoding techniques:  $\MR$, $\ZF$, and  $\PZF$. A key observation from the figure is the enhanced detection probability in single target scenarios compared to multiple targets, which is likely due to the increased interference challenges in multi-target environments. Notably, the $\PZF$ scheme emerges as the most effective technique, achieving the highest detection probability across various SNR levels.  The superiority of $\PZF$ is particularly marked in multiple target scenarios, underlining its efficiency in interference mitigation.

\begin{Remark}
   To improve the detection probability at the radar, the multiple-target detection problem can be modeled as a multi-hypothesis testing, followed by beampattern design at the radar~\cite{Wang:SENS:2018,Yi:TSP:2020}. Exploring the development of a closed-form expression for the detection probability and subsequently applying joint power control and beamforming design at the BS represent interesting future research directions. 
\end{Remark}

\vspace{-2em}
\section{Conclusion}~\label{sec:conc}
We have studied the coexistence of massive MIMO communication and MIMO radar.  $\PZF$ precoding design has been proposed to achieve favourable SE performance in the cellular system, while delivering substantial improvements to the radar system. Closed-form expressions for the SE of users and the detection probability at the radar have been derived for $\MR$, $\ZF$, and $\PZF$ precoding designs. Accordingly, power control at the BS and radar has been developed to manage the inter-system interference, subject to the SE requirements of the users. The optimization problem has been efficiently solved, showing substantial improvement in detection probability, when compared with the equal power allocation scenario.

\section{Appendix}
\subsection{Useful Result}~\label{P:lemmaB}
\vspace{-1em}
\begin{Lemma}~\label{Lemma:B}
For the projection matrix
$\mathbf {B} = \mathbf {I}_{M}-\mathbf{\hat{R}^{H}}\left (\mathbf{\hat{R}}\mathbf{\hat{R}}^{H}\right)^{-1}\mathbf{\hat{R}}$, we have
\vspace{-0.3em}
\begin{align}
    \Ex\{\qB\} = \frac{M-N}{M} \qI_M, \quad M>N.
\end{align}
\end{Lemma}
\vspace{-1em}
\begin{proof}
We denote $\qS=\hR^{H}(\hR\hR^{H})^{-1}\hR$ with $\hR\in\mathbb{C}^{N\times M}$. Assuming that $N < M$, by using the SVD technique, $\qR$ can be expressed as
\vspace{-0.3em}
\begin{align}~\label{eq:SVD:R}
 \hR=\qU\qD\qV^{H}, 
 \end{align}
where $\qU\in \mathbb{C}^{N \times N}$ and $\qV\in \mathbb{C}^{M \times M}$ are unitary matrices and $\qD = [\qD_1, \boldsymbol{0}]\in \mathbb{C}^{N \times M}$ with $\qD_1$ is $N\times N$ nonnegative diagonal matrix.  Therefore, by using~\eqref{eq:SVD:R}, we can express $\qS$ as
\vspace{-0.5em}
\begin{align}~\label{eq:qssvd}
\qS&=\qV\qD^{H}\qU^{H} (\qU\qD\qV^{H}\qV\qD^{H}\qU^{H})^{-1}\qU\qD\qV^{H}
\nonumber\\
&\stackrel{(a)}{=}
\qV\qD^{H}\qU^{H} (\qU\qD^{2}_{1}\qU^{H})^{-1}\qU\qD\qV^{H}
\nonumber\\
&
\stackrel{(b)}{=}\qV\qD^{H}\qD^{-2}_{1}\qD\qV^{H}
\nonumber\\
&=
\qV
\begin{bmatrix}
\qI_N & \boldsymbol{0} \\
\boldsymbol{0} &\boldsymbol{0} 
\end{bmatrix}
\qV^{H},
  \end{align}
where (a) follows from the fact that $\qV^{H}\qV = \qI_M$ and (b) holds since $\qU^{H}\qU  =\qI_N$. 
Now assuming that $\qV \triangleq [\qV_1, \qV_2]$ with $\qV_1 =[\qv_{1,1}, \ldots, \qv_{1,N}] \in \mathbb{C}^{M \times N}$ and $\qV_2 \in \mathbb{C}^{M \times (M-N)}$, we can further simplify~\eqref{eq:qssvd} as   
\vspace{-0.3em}
\begin{align}~\label{eq:qssvd2}
\qS = \sum_{r=1}^{N} \qv_{1,r} \qv_{1,r}^H.
\end{align}
By using~\eqref{eq:qssvd2}, the expectation of $\qS$ can be expressed as
\vspace{-0.2em}
\begin{align}~\label{eq:ExS1}
 \Ex\big\{\qS\big\}=\sum^{K}_{r=1}  \Ex\big\{\qv_{1,r}. \qv_{1,r}^H\big\}.
  \end{align}
Noticing that the distribution of $\qv_{1,r}$ is the same as $\frac{\qz}{\Vert\qz\Vert}$, with $\qz\sim  \mathcal{CN}\left(\boldsymbol{0},\qI_{M} \right)$,~\eqref{eq:ExS1} can be further simplified as
\vspace{-0.3em}
\begin{align}~\label{eq:ExS}
 \Ex\big\{\qS\big\}=N\Ex\bigg\{\frac{\qz\qz^{H}}{\qz^{H}\qz}\bigg\}.
    \end{align}

The diagonal elements of the expectation matrix in~\eqref{eq:ExS} can be obtained as
\vspace{-0.3em}
  \begin{align}
 \left[\Ex\big\{\qS\big\}\right]_{(r,r)}&=N\Ex\{ t\}
 \nonumber\\
&=
\frac{N}{M},
    \end{align}
where  $t = \frac{\vert z_{r}\vert^2}{\sum^{M}_{s=1} \vert z_{s}\vert^2}$ is  Beta distributed random variables with parameter $(1, M-1)$ and $\Ex\{ t\} =\frac{1}{M}$~\cite{peebles2001probability}.

To find the off-diagonal elements of $\Ex\big\{\qS\big\}$, we have~\cite{peebles2001probability}
\vspace{-0.3em}
\begin{align}
 \left[\Ex\big\{\qS\big\}\right]_{(r,s)}&=K\Ex\bigg\{\frac{z_{r}z_{s}^*}{\sum^{M}_{t=1} |z_{t}|^2}\bigg\}
 \nonumber\\
&=
0, \forall r\neq s.
 \end{align}

 To this end, we can conclude that 
\vspace{-0.3em}
  \begin{align}
 \Ex\big\{\qB\big\}&= \qI_M - \Ex\big\{\qS\big\} =\frac{M-N}{M}\qI_M.
    \end{align}
\end{proof}

\subsection{Proof of Proposition~\ref{prop:DP}} \label{app:DP}
To derive the closed-form expression for the $\SINR^{\MR}_k$, we need to compute $\mathrm{DS}_k$, $\vert\mathrm{BU}_k\vert^2$ and $\vert \mathrm{IUI}_{kk'}\vert^2$.
By substituting~\eqref{e:tmr} into the desired signal in~\eqref{eq:DSk}, we obtain firstly  $\mathrm{DS}_k$ as  
\begin{align}~\label{eq:DS:mr:final}
\mathrm{DS}_k &= \sqrt{\rho \eta_k} \Ex\Big\{(\hat{\qg}_{k}^{T} - \tilde{\qg}_{k}^{T}){\qt}^{\MR}_{k}\Big\} \nonumber\\
&=\sqrt{\frac{\rho \eta_k}{M}}\Ex\left\{\hat{\qg}_k^{T}\qh_{k}^{*}\right\}\nonumber\\
&\hspace{0em} =\sqrt{\frac{\rho \eta_k \gamma_k}{M}}\Ex\{\left\|{\qh}_{k}\|^2\right\} =\sqrt{M\rho \gamma_{k}\eta_{k}}.
\end{align}

In order to compute  $\vert\mathrm{BU}_k\vert^2$, we have
 \begin{align}~\label{eq:BU:MR}
 \Ex\Big\{\vert\mathrm{BU}_k\vert^2\Big\} &= \rho \eta_k\Ex\Big\{\vert{\qg}_{k}^{T}{\qt}^{\MR}_{k}\vert^2\Big\}-\vert\mathrm{DS}_k\vert^2\nonumber\\
&=
\rho \gamma_{k}\eta_{k}(M+1)-M\rho \gamma_{k}\eta_{k}\nonumber\\
&=\rho \gamma_{k}\eta_{k}.
 \end{align}
We now turn our attention to derive $\vert \mathrm{IUI}_{kk'}\vert^2$, which is given by
\begin{align}~\label{eq:MRIUI}
\Ex\Big\{ \big\vert \mathrm{IUI}_{kk'} \big\vert^2\Big\}
&=
\rho\eta_{k'}\Ex\Big\{\big|\qg_{k}^{T}\qt_{k'}^{\MR}\big|^2\Big\}\nonumber\\
&=
{\rho\eta_{k'}}\Ex\Big\{\qg_{k}^{T} \Ex\big\{\qt_{k'}^{\MR} (\qt_{k'}^{\MR})^H\big\} \qg_{k}^{*}\Big\}
\nonumber\\
&=\rho\eta_{k'} \gamma_k.
\end{align}
Plugging \eqref{eq:DS:mr:final}, \eqref{eq:BU:MR}, and \eqref{eq:MRIUI} into  ~\eqref {eq:SINE:general}, we obtain the desired result in~\eqref{Prop:SINR:MR}.

\subsection{Proof of Proposition~\ref{prop:DP:ZF}} \label{app:ZZFF:proof}
By using $\qt_k^{\ZF}$, 
i.e., the $k$-th column of $\qT^{\ZF}$ in~\eqref{eq:tzf}, the desired signal in~\eqref{eq:DSk}, can be obtained as
\begin{align}~\label{eq:DS:zf:final}
\mathrm{DS}_k &= \sqrt{\rho \eta_k} \Ex\Big\{(\hat{\qg}_{k}^{T} - \tilde{\qg}_{k}^{T}){\qt}^{\ZF}_{k}\Big\} \nonumber\\
&=\sqrt{\rho \eta_k  }\Ex\left\{\hat{\qg}_k^{T}\qt_k^{\ZF}\right\}\nonumber\\
&\hspace{0em} =\sqrt{\rho \eta_k \gamma_k }\Ex\left\{\qh_k^{T}\qt_k^{\ZF}\right\}\nonumber\\
&\hspace{0em} =\sqrt{(M-K)\rho \eta_k \gamma_k }.
\end{align}
In order to derive $\Ex\{\vert\mathrm{BU}_k\vert^2\}$, which can be further expressed by
 \begin{align}~\label{eq:BU:ZF}
 \Ex\Big\{\vert\mathrm{BU}_k\vert^2\Big\} &={\rho\eta_k}\Ex\Big\{\vert{\qg}_{k}^{T}{\qt}^{\PZF}_{k}\vert^2\Big\}-\vert\mathrm{DS}_k\vert^2,
 \end{align}
 we need to obtain the first term in~\eqref{eq:BU:ZF}, which can be written as
 \begin{align}~\label{eq:Exgktzf}
 \Ex\big\{|\qg_{k}^{T}{\qt}^{\PZF}_{k}|^2\big\}
&= \Ex\Big\{\qg_{k}^{T}\qt_{k}^{\ZF}(\qt^{{\ZF}}_{k})^H 
  \qg_{k}^{*}\Big\}.
 \end{align}
 
Denote {$\qV^{\ZF}_k=\qt_{k}^{\ZF}(\qt^{\ZF}_{k})^H$}, we can express~\eqref{eq:Exgktzf} as
 \begin{align}~\label{eq:EX:ZFgVg}
\Ex\big\{|\qg_{k}^{T}{\qt}^{\ZF}_{k}|^2\big\}
 &=
\Ex\big\{(\hat{\qg}_{k}^{T} - \tilde{\qg}_{k}^{T} )\qV^{\ZF}_k (\hat{\qg}_{k}^{*} - \tilde{\qg}_{k}^{*} ) \big\}
  \nonumber\\
  &=
  \Ex\big\{(\hat{\qg}_{k}^{T}\qV^{\ZF}_k\hat{\qg}_{k}^{*}+ \tilde{\qg}_{k}^{T}\qV^{\ZF}_k\tilde{\qg}_{k}^{*} ) \big\}
  \nonumber\\
  &=
  \Ex\big\{\hat{\qg}_{k}^{T}\qV^{\ZF}_k\hat{\qg}_{k}^*\big\}+ \Ex\big\{\tilde{\qg}_{k}^{T}\qV^{\ZF}_k\tilde{\qg}_{k}^{*}\big\}
    \nonumber\\
  &=
  (M-K)\rho \eta_k \gamma_k+  (\beta_{k}-\gamma_{k}).
   \end{align}
To this end, by substituting~\eqref{eq:EX:ZFgVg} into~\eqref{eq:BU:ZF}, we have
\begin{align}~\label{eq:BU:ZF}
 \Ex\Big\{\vert\mathrm{BU}_k\vert^2\Big\} &= \Ex\Big\{\vert{\qg}_{k}^{T}{\qt}^{\ZF}_{k}\vert^2\Big\}-\vert\mathrm{DS}_k\vert^2
 \nonumber\\
 &=
 \rho \eta_k(\beta_{k}-\gamma_{k}).
 \end{align}

Then, we obtain the inter-user interference term as
\begin{align}~\label{eq:IUI:ZF}
\Ex\Big\{ \big\vert \mathrm{IUI}_{kk'} \big\vert^2\Big\}
&=
\rho\eta_{k'}\Ex\Big\{\big|\qg_{k}^{T}\qt_{k'}^{\ZF}\big|^2\Big\}\nonumber\\
&=
{\rho\eta_{k'}}\Ex\Big\{\tilde{\qg}_{k}^{T} \Ex\big\{\qt_{k'}^{\ZF} (\qt_{k'}^{\ZF})^H\big\} \tilde{\qg}_{k}^{*}\Big\}
\nonumber\\
&=\rho\eta_{k'} (\beta_k-\gamma_k).
\end{align}
Finally,~\eqref{Prop:SINR:ZF} is obtained by plugging ~\eqref{eq:DS:zf:final},~\eqref{eq:Exgktzf}, and~\eqref{eq:IUI:ZF} into~\eqref {eq:SINE:general}.

\subsection{Proof of Proposition~\ref{prop:DP:PZF}} \label{app:PZF:proof}
By using $\qt^{\PZF}$ in~\eqref{P:PZFe}, the desired signal in~\eqref{eq:DSk}, can be obtained as
\begin{align}~\label{eq:DS:pzf:final}
\mathrm{DS}_k &= \sqrt{\rho \eta_k} \Ex\Big\{(\hat{\qg}_{k}^{T} - \tilde{\qg}_{k}^{T}){\qt}^{\PZF}_{k}\Big\} \nonumber\\
 &\hspace{0em} =\alpha_{\PZF}\sqrt{\rho \eta_k \gamma_k }\Ex\big\{\qh_k^{T}\Ex\{\qB\}\qw^{\ZF}_{k}\big\}\nonumber\\
&\hspace{0em} \stackrel{(a)}{=}\alpha_{\PZF}\sqrt{\rho\eta_k \gamma_k }
\bigg(\frac{M-N}{M}\bigg)\Ex\left\{\qh_k^{T}{\qw}^{\ZF}_{k}\right\}\nonumber\\
 &\hspace{0em} \stackrel{(b)}{=}\sqrt{ \frac{\rho\eta_k \gamma_k(M-K)(M-N)}{M}    },
\end{align}
where we exploited, in (a) Lemma~\ref{Lemma:B} and in (b) $\Ex\left\{\qh_k^{T}{\qw}^{\ZF}_{k}\right\}=1$.

In order to derive $\Ex\{\vert\mathrm{BU}_k\vert^2\}$, which can be expressed as
 \begin{align}~\label{eq:BU:PZF}
 \Ex\Big\{\big\vert\mathrm{BU}_k\big\vert^2\Big\} &= {\rho\eta_k}\Ex\Big\{\big\vert{\qg}_{k}^{T}{\qt}^{\PZF}_{k}\big\vert^2\Big\}-\vert\mathrm{DS}_k\vert^2,
 \end{align}
 we need to obtain the first term in~\eqref{eq:BU:PZF}, given by
 \begin{align}~\label{eq:Exgktpzf}
 \Ex\Big\{\big|\qg_{k}^{T}{\qt}^{\PZF}_{k}\big|^2\Big\}
&= \Ex\Big\{\qg_{k}^{T}\qB\qw_{k}^{\ZF}(\qw^{{\ZF}}_{k})^H\qB^{H} 
  \qg_{k}^{*}\Big\}.
 \end{align}
 
Denoting {$\qV_k=\qw_{k}^{\ZF}(\qw^{\ZF}_{k})^H$}, we can express~\eqref{eq:Exgktpzf} as
 \begin{align}~\label{eq:gkwkzf}
\Ex\Big\{\big|\qg_{k}^{T}{\qt}^{\PZF}_{k}\big|^2\Big\}
&=
 \Ex\left\{\qg_{k}^{T}\qB\qV_k\qB^{H} \qg_{k}^{*}\right\}\nonumber\\
&= \Ex\Big\{\qg_{k}^{T}\big(\qI_{M}-\hR^{H}( \hR\hR^{H})^{-1} \hR\big) 
   \qV_k \nonumber\\
&\hspace{1em}\times\big(\qI_{M}-\hR^{H}( \hR\hR^{H})^{-1} \hR\big)\qg_{k}^{*}\Big\}
\nonumber\\
&= \Ex\Big\{\qg_{k}^{T}\qV_k \qg_{k}^{*}-2\qg_{k}^{T}\qV_k\hR^{H}( \hR\hR^{H})^{-1} \nonumber\\
&\hspace{1em}
\times\hR\qg_{k}^{*}+ \qg_{k}^{T}\hR^{H}( \hR\hR^{H})^{-1} \hR
\nonumber\\
&\hspace{1em}\times\qV_k(\hR^{H}( \hR\hR^{H})^{-1} \hR\qg_{k}^{*}\Big\}.
\end{align}
Now, by using the fact that $\hR$ is independent of $\qg_k$ and $\qV_k$, we can derive~\eqref{eq:gkwkzf} as
\begin{align}~\label{eq:gkwkzf2}
&\Ex\Big\{\big|\qg_{k}^{T}{\qt}^{\PZF}_{k}\big|^2\Big\}
=\Ex\Big\{\qg_{k}^{T}\qV_k \qg_{k}^{*}\Big\}
\nonumber\\
&\hspace{0.5em}-
2\Ex\Big\{\qg_{k}^{T}\qV_k \Ex\Big\{ \hR^{H}(\hR\hR^{H})^{-1} \hR\Big\}\qg_{k}^{*}\Big\}
\nonumber\\
&\hspace{0.5em}+ \Ex\Big\{\qg_{k}^{T}
\Ex\Big\{\hR^{H}( \hR\hR^{H})^{-1} \hR\qV_k(\hR^{H}( \hR\hR^{H})^{-1} \hR\big\}\qg_{k}^{*}\Big\}
\nonumber\\
&\stackrel{(a)}{=}\!\Ex\big\{\qg_{k}^{T}\qV_k \qg_{k}^{*}\big\}\!-\!
2\frac{N}{M}\Ex\big\{\qg_{k}^{T}\qV_k \qg_{k}^{*}\big\}\!+ \!\bigg(\frac{N}{M}\bigg)^2\Ex\big\{\qg_{k}^{T}\qV_k\qg_{k}^{*}\Big\}\nonumber\\
&=
\Big(1- \frac{N}{M}\Big)^2\Ex\big\{\qg_{k}^{T}\qV_k \qg_{k}^{*}\big\},
 \end{align}
where (a)  holds since according to Lemma~\ref{Lemma:B} we have $\Ex\left\{\hR^{H}(\hR\hR^{H})^{-1}\hR\right\}=\frac{N}{M}\qI_M$,  while the final result follows from the fact that $\Ex\big\{\qg_{k}^{T}\qV_k \qg_{k}^{*}\big\}=1$. 

Then, we need to derive $\Ex\big\{\qg_{k}^{T}\qV_k \qg_{k}^{*}\big\}$, which after vanishing the cross-expectations can be obtained as
\begin{align}~\label{eq:EX:gVg}
  \Ex\big\{\qg_{k}^{T}\qV_k \qg_{k}^{*}\big\} &=\Ex\big\{(\hat{\qg}_{k}^{T} - \tilde{\qg}_{k}^{T} )\qV_k (\hat{\qg}_{k}^{*} - \tilde{\qg}_{k}^{*} ) \big\}
  \nonumber\\
   &=
  \Ex\big\{\hat{\qg}_{k}^{T}\qV_k\hat{\qg}_{k}^*\big\}+ \Ex\big\{\tilde{\qg}_{k}^{T}\qV_k\tilde{\qg}_{k}^{*}\big\}
    \nonumber\\
  &=
  1+  (\beta_{k}-\gamma_{k}).
  \end{align}

To this end, by substituting~\eqref{eq:EX:gVg} into~\eqref{eq:gkwkzf2} and then plugging the result into~\eqref{eq:BU:PZF}, we get
\begin{align}~\label{eq:BU:pzf:final}
 \Ex\Big\{\big\vert\mathrm{BU}_k\big\vert^2\Big\} &= \Ex\Big\{\big\vert{\qg}_{k}^{T}{\qt}^{\PZF}_{k}\big\vert^2\Big\}-\vert\mathrm{DS}_k\vert^2\nonumber\\
 &={\rho \eta_k}\alpha_{\PZF}^2
 \Big(1- \frac{N}{M}\Big)^2 (\beta_{k}-\gamma_{k}).
 \end{align}

In order to derive IUI in~\eqref{eq:IUIk}, by using similar steps as in~\eqref{eq:gkwkzf}, we have
 \begin{align}~\label{eq:IUI:pzf}
\Ex\Big\{ \big\vert \mathrm{IUI}_{kk'} \big\vert^2\Big\}
&=
{\rho\eta_{k'}}\Ex\Big\{\big|\qg_{k}^{T}\qB\qw_{k'}^{\ZF}\big|^2\Big\}
\nonumber\\
&=
{\rho\eta_{k'}}\Big(\Ex\big\{\qg_{k}^{T}\qV_{k'} \qg_{k}^{*}\big\}
-
2\frac{N}{M}\Ex\big\{\qg_{k}^{T}\qV_{k'} \qg_{k}^{*}\big\}
\nonumber\\
&\hspace{1em}
+ \bigg(\frac{N}{M}\bigg)^2
\Ex\big\{\qg_{k}^{T}\qV_{k'}\qg_{k}^{*}\big\}\Big)
\nonumber\\
&\hspace{0em}=
 {\rho\eta_{k'}}\Big(1- \frac{N}{M}\Big)^2\Ex\big\{\qg_{k}^{T}\qV_{k'} \qg_{k}^{*}\big\}.
 \end{align}
Therefore, we need to derive
$\Ex\big\{\qg_{k}^{T}\qV_{k'} \qg_{k}^{*}\big\}$, which, after averaging out the cross-expectations, it becomes
\begin{align}
  \Ex\big\{\qg_{k}^{T}\qV_{k'} \qg_{k}^{*}\big\} &=\Ex\big\{(\hat{\qg}_{k}^{T} - \tilde{\qg}_{k}^{T} )\qV_{k'} (\hat{\qg}_{k}^{*} - \tilde{\qg}_{k}^{*} ) \big\}
  \nonumber\\
  &=
  \Ex\big\{(\hat{\qg}_{k}^{T}\qV_{k'}\hat{\qg}_{k}^{*}+ \tilde{\qg}_{k}^{T}\qV_{k'}\tilde{\qg}_{k}^{*} ) \big\}
  \nonumber\\
  &=
  \Ex\big\{\hat{\qg}_{k}^{T}\qV_{k'}\hat{\qg}_{k}^*\big\}+ \Ex\big\{\tilde{\qg}_{k}^{T}\qV_{k'}\tilde{\qg}_{k}^{*}\big\}
    \nonumber\\
  &=
  (\beta_{k}-\gamma_{k}).
  \end{align}
  Therefore, the IUI in~\eqref{eq:IUI:pzf}, can be derived as
 \begin{align}~\label{eq:IUI:pzf:final}
\Ex\Big\{ \big\vert \mathrm{IUI}_{kk'} \big\vert^2\Big\}
=
 {\rho\eta_{k'}}\Big(1- \frac{N}{M}\Big)^2(\beta_{k}-\gamma_{k}).
 \end{align}
Finally, by substituting~\eqref{eq:DS:pzf:final},~\eqref{eq:BU:pzf:final}, and~\eqref{eq:IUI:pzf:final}, into~\eqref{eq:SINE:general}, the desired result in~\eqref{Prop:SINR:PZF} is obtained.

\subsection{Proof of Proposition~\ref{prop:pd:PZF}} \label{app:PZF:pd}
From \eqref{e:tpzf}, $\tilde{\qT}^{\PZF}$ can be expressed as
\begin{align}~\label{eq:tildTPZF:mat} 
\tilde{\qT}^{\PZF} &= \qT^{\PZF}\qD_{\eta}(\qT^{\PZF})^H\nonumber\\
&=\frac{M(M-K)}{(M-N)}\mathbf{B} \mathbf{W}^{\ZF}
\qD_{\eta}(\mathbf{W}^{\ZF})^{H}\mathbf{B}^{H}.
\end{align}

Therefore, $\trac(\tilde{\qT}^{\PZF} )$ can be obtained as
\begin{align}~\label{eq:trtildTPZF} 
\trac(\tilde{\qT}^{\PZF}) &
\stackrel{(a)}{=}\!\frac{M(M\!-\!K)}{(M\!-\!N)}\trac( \mathbf{W}^{\ZF}\qD_{\eta}
(\qW^{\ZF})^{H}\qB^H\qB)
\nonumber\\
&\hspace{0em}\stackrel{(b)}{=}\frac{M(M-K)}{(M-N)}\nonumber\\
&\hspace{0em}\times\trac(\qH^{*}( {\qH}^{T} {\qH}^{*})^{-1} \qD_{\eta} ( {\qH}^{T} {\qH}^{*})^{-1}\qH^{T}\qB),
\end{align}
where we have exploited: in (a) $\trac(\qX\qY) = \trac(\qY\qX)$, in (b) $\qB^H\qB =\qB$.
To derive $\trac(\tilde{\qT}^{\PZF})$, let us denote $\qC$ as 
\begin{align} \label{P:C}
\qC =
\qH^{*}( {\qH}^{T} {\qH}^{*})^{-1}
\qD_{\eta}
( {\qH}^{T} {\qH}^{*})^{-1}\qH^{T}.
\end{align}
Then, by substituting \eqref{P:B} into \eqref{eq:trtildTPZF},   we get 
\begin{align} \label{P:T1T2}
\trac(\tilde{\qT}^{\PZF})&=\frac{M(M-K)}{(M-N)}\nonumber\\
&\hspace{0em}\times
\Big(\underbrace{\trac(\qC)}_{\text{T}_1}-\underbrace{\trac(\qC\mathbf{\hat{R}^{H}}\left (\mathbf{\hat{R}}\mathbf{\hat{R}}^{H}\right)^{-1}\mathbf{\hat{R}})}_{\text{T}_2}\Big).
\end{align}
The first term $\text{T}_1$ can be calculated as 
\begin{align}~\label{eq:T1:final}
T_1&=
\trac   ( \qH^{*}( {\qH}^{T} {\qH}^{*})^{-1} \qD_{\eta}
 ({\qH}^{T} {\qH}^{*})^{-1}\qH^{T})\nonumber\\
&\stackrel{(a)}{=}
\trac( \qD_{\eta}({\qH}^{T} {\qH}^{*})^{-1})\nonumber\\
&\stackrel{(b)}{\approx}
\frac{1}{M}\trac(\qD_{\eta}\qI_K) \nonumber\\
& =\frac{1}{M} \sum_{k=1}^{K}\eta_{k}, 
\end{align}
where we have exploited: in (a) $\trac(\qX\qY) = \trac(\qY\qX)$, in (b) $({\qH}^{T} {\qH}^{*})^{-1}\approx\frac{1}{M}\qI_K$.

Before proceeding to derive $T_2$, we notice that 
\begin{align}
\qC& =
\qH^{*}( {\qH}^{T} {\qH}^{*})^{-1}
\qD_{\eta}
( {\qH}^{T} {\qH}^{*})^{-1}\qH^{T}
\nonumber\\
 &\approx
 \qH^{*}(M\qI_{K})^{-1}
 \qD_{\eta}
 (M\qI_{K})^{-1}\qH^{T}\nonumber\\
 &=
 \frac{1}{M^2}\qH^{*}
 \qD_{\eta}
 \qH^{T}.
\end{align}
Therefore, we have 
\vspace{-0.5em}
\begin{align}\label{eq:T2:def2}
T_2&=\trac\big(\qC {\hat{\qR}}^{H}\left (\mathbf{\hat{R}}\mathbf{\hat{R}^{H}}\right)^{-1}\mathbf{\hat{R}}\big)\nonumber\\
&\approx
\frac{1}{M^2}\trac(\qD_{\eta}\qH^{T}{\hat{\qR}}^{H}\left (\mathbf{\hat{R}}\mathbf{\hat{R}^{H}}\right)^{-1}\mathbf{\hat{R}}\qH^{*}).
\end{align}

Accordingly, by using the law of large number when $M\to\infty$, we get
\vspace{-0.9em}
\begin{align}~\label{eq:aympt}
&\frac{\qH^{T}{\hat{\qR}}^{H}\left (\mathbf{\hat{R}}\mathbf{\hat{R}^{H}}\right)^{-1}\mathbf{\hat{R}}\qH^{*}}{M}\nonumber\\
&-\underbrace{\frac{\trac({\hat{\qR}}^{H}\left (\mathbf{\hat{R}}\mathbf{\hat{R}^{H}}\right)^{-1}\mathbf{\hat{R}})}{M}}_{=N/M}\qI_{K} &\stackrel{M\to\infty}{\rightarrow } 0.
\end{align}

Finally, by substituting~\eqref{eq:aympt} into~\eqref{eq:T2:def2}, $T_2$ can be approximated as
\vspace{-0.3em}
\begin{align}~\label{eq:T2:final}
T_2&\approx
\frac{N}{M^2}\trac(\qD_{\eta}\mathbf)\nonumber\\
&=
\frac{N}{M^2}  \sum_{k=1}^{K}\eta_{k}. 
\end{align}
To this end, by substituting~\eqref{eq:T1:final} and~\eqref{eq:T2:final} into~\eqref{P:T1T2} and then plugging the results into~\eqref{eq:lmu}, the desired result in~\eqref{P:muPZF} is obtained.

\begin{figure*}
\begin{align}~\label{eq:Ex:final}
\Ex\Big\{ \big\vert 
   \qg_k^T\qt^{\MR}_{k'} \big\vert^2\Big\} &=
   \begin{cases}
       \frac{\tau_{\mathrm {p}} \rho_u}{\trac(\qQ_{k'})}  \Big( {\tau_{\mathrm {p}} {\rho_{\mathrm{\text p}}}}
    \Big\vert \trac\big( \boldsymbol{\Omega}_{k'}\boldsymbol{\Omega}_k{\boldsymbol{\Xi }_{k}}
     \big) \Big\vert^2
      +     
     \trac\big(  \boldsymbol{\Omega}_{k'}
     \boldsymbol{\Omega}_k\boldsymbol{\Xi }_k \boldsymbol{\Omega}_k
     \big)\Big)& k' = k
       \\
       \frac{\tau_{\mathrm {p}} \rho_u}{\trac(\qQ_{k'})} \trac( \boldsymbol{\Omega}_{k'}{{\boldsymbol{\Xi }}_{k'}}\boldsymbol{\Omega}_{k'}\boldsymbol{\Omega}_k) & k'\neq k
   \end{cases}\tag{103}
\end{align}   
	\hrulefill
	\vspace{-1mm}
\end{figure*}

\subsection{Proof of Proposition~\ref{propos:MR:correlated}} \label{proof:propos:MR:correlated}
An achievable downlink SE at the $k$-th user can be obtained using~\eqref{eq:dLSE}
where the effective SINR is given by
\begin{align}~\label{eq:SINE:general:CR}
    &\SINR_k=\nonumber\\
    &
    \!\frac{
                {{\rho\eta_k} } \big\vert   
 \Ex\big\{\qg_k^T\qt_k \big\}  \big\vert^2
                 }
                 {
                 \sum_{k'=1}^K\!\!{{\rho\eta_{k'}}  }
                  \Ex\Big\{ \big\vert 
   \qg_k^T\qt_{k'} \big\vert^2\Big\}
                  \!-\! {\rho\eta_{k}}\Big\vert\Ex\big\{\qg_k^T\qt_k \big\}  \Big\vert^2                 
                  \! +\!  P_R\bar{\beta}_k N +  \sigma_{C}^2}.
\end{align}

Noticing that the MR precoding is given $\qt^{\MR}_k = \frac{\hat{\qg}_k^*}{\sqrt{\Ex\{\Vert \hat{\qg}_k \Vert^2}\}}$, where $\Ex\{\Vert \hat{\qg}_k \Vert^2\}=\trac(\qQ_k)$, we have
\begin{align}~\label{eq:DS:mr:finalcor}
\Ex\big\{\qg_k^T\qt^{\MR}_k \big\} &=  \Ex\Big\{(\hat{\qg}_{k}^{T} + \tilde{\qg}_{k}^{T}){\qt}^{\MR}_{k}\Big\} \nonumber\\
&=\frac{1}{\sqrt{\trac(\qQ_k)}}\Ex\Big\{\Vert \hat{\qg}_{k}\Vert^2\Big\}\nonumber\\
&\hspace{0em}  =\sqrt{\trac(\qQ_k)} = \sqrt{\tau_{\mathrm {p}} \rho_u \trac(\boldsymbol{\Omega}_{k} \boldsymbol{\Xi }_{k}\boldsymbol{\Omega}_{k})}.
\end{align}
To derive $\Ex\Big\{ \big\vert 
   \qg_k^T\qt^{\MR}_{k'} \big\vert^2\Big\}$, we consider two different cases as follows
\begin{itemize}
\item When $k' \neq k$:
\begin{align}
  \Ex\big\{
\vert 
   {\qg}_{k}^{T}\hat{\qg}_{k'}^*\vert^2  \big\} 
   &=\Ex\{ {\qg}_{k}^{T}\hat{\qg}_{k'}^*\hat{\qg}_{k'}^T {\qg}_{k}^{*}\}
   \nonumber\\
&= \trac \big(\Ex\{ \hat{\qg}_{k'}^*\hat{\qg}_{k'}^T {\qg}_{k}^{*}{\qg}_{k}^{T} \}\big) \nonumber\\
&=\trac\big( \Ex\{ \hat{\qg}_{k'}^*\hat{\qg}_{k'}^{T}\}\Ex\{{\qg}_{k}^*{\qg}_{k}^{T} \}\big)
\nonumber\\
&= \tau_{\mathrm {p}} \rho_u\trac( \boldsymbol{\Omega}_{k'}{{\boldsymbol{\Xi }}_{k'}}\boldsymbol{\Omega}_{k'}\boldsymbol{\Omega}_k).
\end{align}

\item When $k'= k$: We notice that according to~\eqref{eq:ykp} and~\eqref{eq:Ek}, we have $\ykp = \ykpp$ and $\boldsymbol{\Xi }_{k}=\boldsymbol{\Xi }_{k'}$, respectively. We first rewrite $\Ex\Big\{ \big\vert 
   \qg_k^T\qt^{\MR}_{k'} \big\vert^2\Big\}$ as
\begin{align}~\label{eq:Ex:sfinal}
&\Ex\Big\{ \big\vert 
   \qg_k^T\qt^{\MR}_{k'} \big\vert^2\Big\} = \frac{1}{\trac(\qQ_{k'})}\Ex\big\{\vert 
   {\qg}_{k}^{T}\hat{\qg}_{k'}^*\vert^2\big\}\nonumber\\
   &= \frac{1}{\trac(\qQ_{k'})}\Ex\big\{
   (\hat{\qg}_{k}^{T}\hat{\qg}_{k'}^*\hat{\qg}_{k'}^{T}\hat{\qg}_{k}^* 
   +
   \tilde{\qg}_{k}^{T}\hat{\qg}_{k'}^* \hat{\qg}_{k'}^{T}\tilde{\qg}_{k}^*)
   \big\},
   \end{align}
where the cross-expectations are vanished as the channel estimation error is zero-mean vector. Now, we derive the first term in~\eqref{eq:Ex:sfinal}, which can be expressed as
\begin{align}~\label{eq:Eqkqkp}
    \Ex\big\{   \hat{\qg}_{k}^{T}\hat{\qg}_{k'}^*\hat{\qg}_{k'}^{T}\hat{\qg}_{k}^*   \big\}
     &=  {\tau_{\mathrm {p}}^2 {\rho^2 _{\mathrm{\text p}}}} 
    \Ex\Big\{
     \Big \vert \ykp^T {\boldsymbol{\Xi }_{k}}\boldsymbol{\Omega}_{k'}
     \boldsymbol{\Omega}_k\boldsymbol{\Xi }_k\ykp^*\Big \vert^2\Big\}.
    \end{align}

Before proceeding, we present the following results~\cite{Matthaiou:TWC:2015}
\begin{align} ~\label{eq:lemma:aBa}
\mathbb {E} \left \{{ {{{\left |\qu^H\qB\qu\right |}^{2}}} }\right \} = |{\mathop {\text {tr}}\nolimits } ({ {\qB \boldsymbol{\Psi}}}){|^{2}} + {\mathop {\text {tr}}\nolimits } \left ({{{ {\mathbf {B}} \boldsymbol{\Psi}}{{ {\mathbf B}}^{H}}{ {\boldsymbol{\Psi}}}} }\right),\end{align}
where $\qu\sim\mathcal{CN}(\boldsymbol{0},\boldsymbol{\Psi})$ with covariance matrix $\boldsymbol{\Psi}\in\mathbb{C}^{M\times M}$ and any diagonalizable matrix $\qB\in\mathbb{C}^{M\times M}$. 

By exploiting~\eqref{eq:lemma:aBa} and noticing that $\ykp\sim\mathcal{CN}(\boldsymbol{0}, ({\boldsymbol{\Xi }_{k}})^{-1})$, we can derive~\eqref{eq:Eqkqkp} as
\begin{align}~\label{eq:ghat}
 \Ex\big\{   \hat{\qg}_{k}^{T}\hat{\qg}_{k'}^*\hat{\qg}_{k'}^{T}\hat{\qg}_{k}^*
   \big\} 
   &= {\tau_{\mathrm {p}}^2 {\rho^2_{\mathrm{\text p}}}}
   \Big( \Big\vert \trac\big( {\boldsymbol{\Xi }_{k}}\boldsymbol{\Omega}_{k'}
     \boldsymbol{\Omega}_k\boldsymbol{\Xi }_k ({\boldsymbol{\Xi }_{k}})^{-1}\big) \Big\vert^2
     \nonumber\\
     &\hspace{-5em}
     +
     \trac\big(  {\boldsymbol{\Xi }_{k}}\boldsymbol{\Omega}_{k'}
     \boldsymbol{\Omega}_k\boldsymbol{\Xi }_k ({\boldsymbol{\Xi }_{k}})^{-1}     
     {\boldsymbol{\Xi }_{k}}\boldsymbol{\Omega}_k
     \boldsymbol{\Omega}_{k'}\boldsymbol{\Xi }_k ({\boldsymbol{\Xi }_{k}})^{-1}\big)\Big)
     \nonumber\\
     &\hspace{-5em}= 
     {\tau_{\mathrm {p}}^2 {\rho^2_{\mathrm{\text p}}}}
    \Big\vert \trac\big( {\boldsymbol{\Xi }_{k}}\boldsymbol{\Omega}_{k'}
     \boldsymbol{\Omega}_k\big) \Big\vert^2
     \nonumber\\
      &\hspace{-5em}+
      {\tau_{\mathrm {p}} {\rho_{\mathrm{\text p}}}}
     \trac\big(  (\boldsymbol{\Omega}_{k'}-\qC_{k'})
     \boldsymbol{\Omega}_k\boldsymbol{\Xi }_k \boldsymbol{\Omega}_k
     \big),
\end{align}
where in the last step we used the fact that $\trac(\qA\qB) =\trac(\qB\qA)$ and applied~\eqref{eq:Ck}.
We now turn our attention to the second term in~\eqref{eq:Ex:sfinal}, which can be expressed as
\begin{align}~\label{eq:gtild}
    \Ex\big\{
   \tilde{\qg}_{k}^{T}\hat{\qg}_{k'}^* \hat{\qg}_{k'}^{T}\tilde{\qg}_{k}^*
   \big\} &= \trac \big(
   \Ex\big\{\tilde{\qg}_{k}^*
   \tilde{\qg}_{k}^{T}\}
   \Ex\{\hat{\qg}_{k'}^*\hat{\qg}_{k'}^{T}
   \big\}\big)
   \nonumber\\
   &= {\tau_{\mathrm {p}} {\rho_{\mathrm{\text p}}}}
   \trac(\qC_{k'}\boldsymbol{\Omega}_k\boldsymbol{\Xi }_k \boldsymbol{\Omega}_k).
\end{align}
\end{itemize}
To this end, by substituting~\eqref{eq:ghat} and~\eqref{eq:gtild}  into~\eqref{eq:Ex:sfinal}, we obtain $\Ex\Big\{ \big\vert 
\qg_k^T\qt^{\MR}_{k'} \big\vert^2\Big\}$ as~\eqref{eq:Ex:final} at the top of the page. Finally, by substituting~\eqref{eq:DS:mr:finalcor} and~\eqref{eq:Ex:final} into~\eqref{eq:SINE:general:CR} and then plugging the result into~\eqref{eq:dLSE}, the desired result in~\eqref{eq:SINE:MR:correlatedPC} is obtained. 

To derive the detection probability over spatially correlated fading channels, according to~\eqref{eq:mu:general}, we need to obtain $\trac(\tilde{\qT}^\MR)$, which can be expressed as
\begin{align}~\label{eq:traTMRcor1}
\trac(\tilde{\qT}^{\MR})&=  
 \sum_{k=1}^{K}
 \frac{\eta_k}{\trac(\qQ_k)}
\trac({\hat{\qg}_k^*} {\hat{\qg}_k^T})
\nonumber\\
&= \sum_{k=1}^{K}
 \frac{\eta_k}{\trac(\qQ_k)}
\Vert \hat{\qg}_k \Vert^2.\tag{104}
\end{align}
We notice that $\hat {\qg}_{k}\sim\mathcal {CN}\left (\pmb {0},\;\qQ_{k} \right)$ can be represented as
\begin{align}
   \hat {\qg}_{k} = \qQ_{k}^{\frac{1}{2}} \qz_k,\tag{105}
\end{align}
where ${\qz}_{k}\sim\mathcal {CN}\left (\pmb {0},\;\qI_M \right)$. Therefore, $\Vert \hat{\qg}_k \Vert^2$ in~\eqref{eq:traTMRcor1} can be expressed as
\begin{align}~\label{eq:normgkcor}
   \Vert \hat{\qg}_k \Vert^2 = \qz_k^H \qQ_k\qz_k.\tag{106}  
\end{align}
By substituting~\eqref{eq:normgkcor} into~\eqref{eq:traTMRcor1} and then applying the trace lemma, i.e., $\qz_k^H \qQ_k\qz_k\approx \trac(\qQ_k)$ as $M$ is large, we get 
\begin{align}~\label{eq:traTMRcor}
\trac(\tilde{\qT}^{\MR}) 
&\approx \sum_{k=1}^{K}
 \eta_k.\tag{107}
\end{align}

To this end,  by substituting~\eqref{eq:traTMRcor} to~\eqref{eq:mu:general}, the non-centrality parameter for MR processing and over spatially correlated fading channels is obtained as~\eqref{eq:lmu2:MRcor}.

\bibliographystyle{IEEEtran}
\bibliography{IEEEabrv,ref_bib}

\begin{IEEEbiography}[{\includegraphics[width=1in,height=1.25in,clip,keepaspectratio]{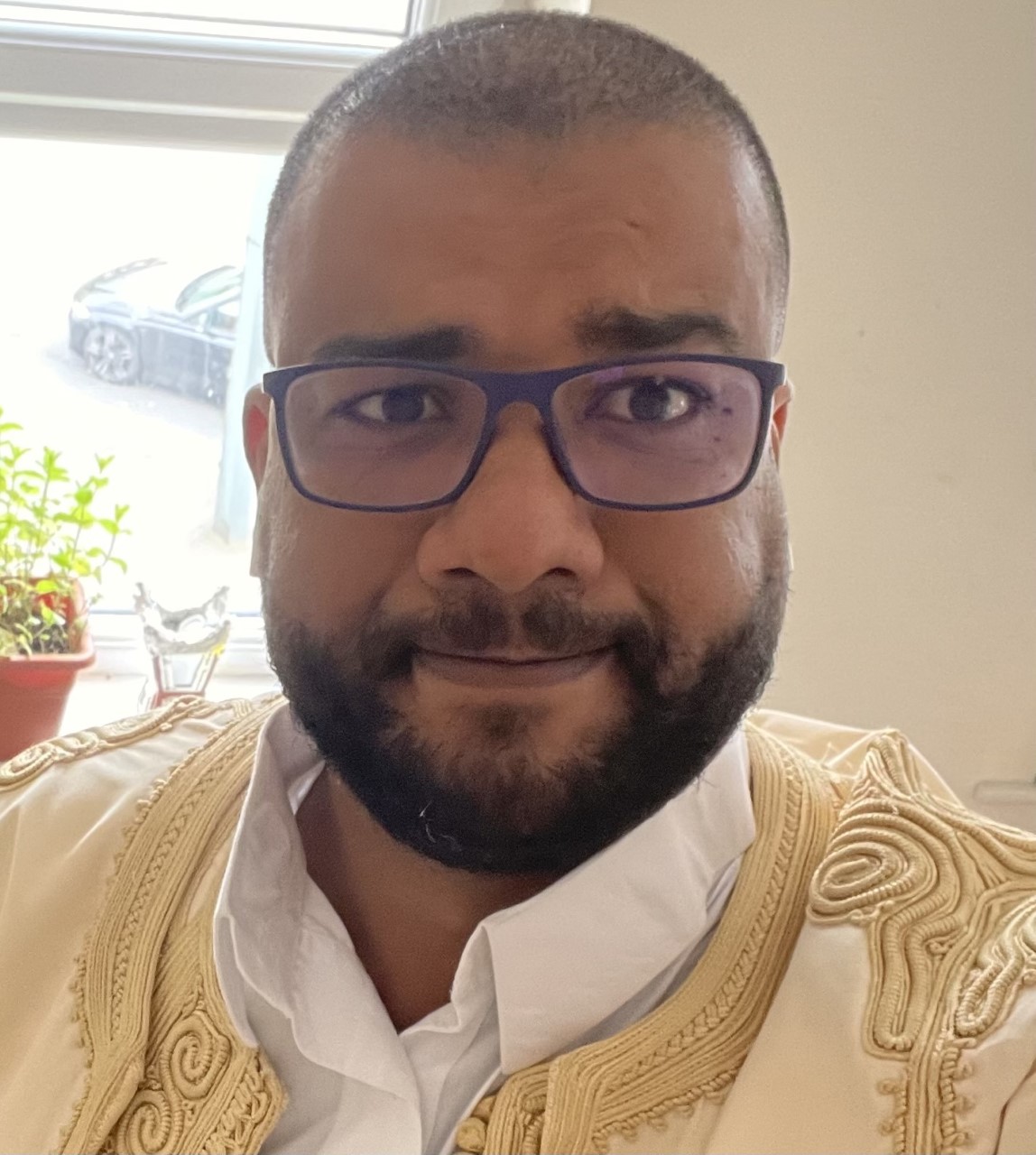}}]
{Mohamed Elfiatoure}~received his B.S. Degree in Electrical and Electronic Engineering from the High Institute Comprehensive Al-Shati, Libya, in 2008, and his M.S. Degree in Electrical and Electronic Engineering from Sheffield Hallam University, U.K., in 2017. 
He is currently pursuing his Ph.D. at the Centre for Wireless Innovation, Queen's University Belfast, U.K.
His research focuses on the coexistence of massive MIMO and MIMO radar, with particular emphasis on multiple-antenna users and precoding techniques.

\end{IEEEbiography}

\begin{IEEEbiography}[{\includegraphics[width=1in,height=1.25in,clip,keepaspectratio]
{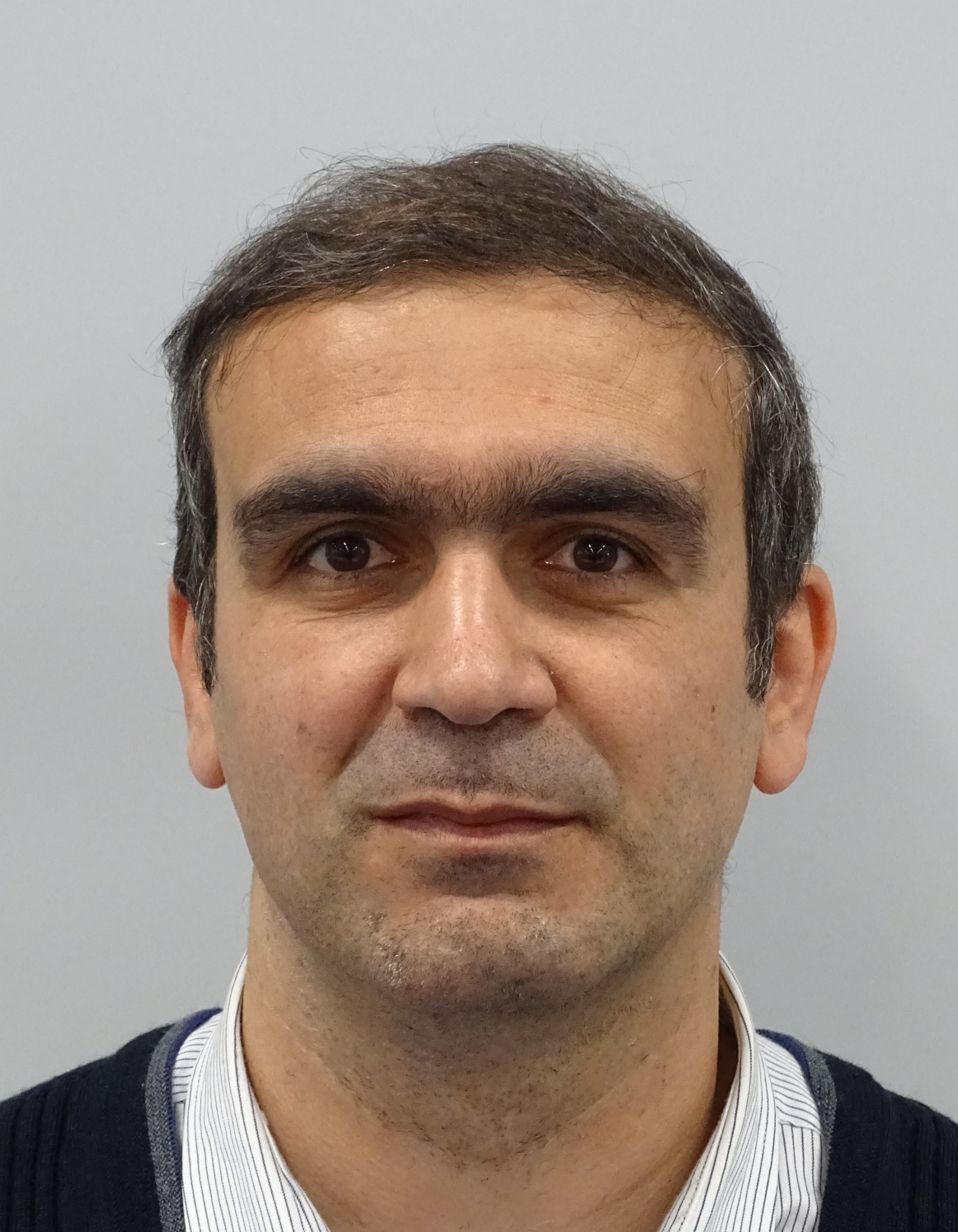}}]
{Mohammadali Mohammadi}~(M'15, SM'23)  received the B.Sc. degree in electrical engineering from the Isfahan University of Technology, 
Isfahan, Iran, the M.Sc. and Ph.D. degree in wireless communication engineering from K. N. Toosi University of Technology, Tehran, Iran in 2007 and 2012, respectively. From November 2010 to  November 2011, he was a visiting researcher in the Research School of Engineering, the Australian National University, Australia. In 2013, he joined the Shahrekord University, Iran, as an Assistant Professor, and was promoted to Associate Professor before joining Queen's University Belfast in 2021 as a Research Fellow. His research interests span diverse areas, such as full-duplex communications, wireless power transfer, OTFS modulation, reconfigurable intelligent surface, and cell-free massive MIMO. He has published more than 70 research papers in accredited international peer reviewed journals and conferences in the area of wireless communication. He has co-authored two book chapters, ``Full-Duplex Non-orthogonal Multiple Access Systems", invited chapter in Full-Duplex Communication for Future Networks, Springer-Verlag, 2020 and ``Full-Duplex wireless-powered communications", invited chapter in Wireless Information and Power Transfer: A New Green Communications Paradigm, Springer-Verlag, 2017. He was a recipient of the Exemplary Reviewer Award for IEEE Transactions on Communications, in 2020 and 2022 and IEEE Communications Letters in 2023. He has actively served as the Technical Program Committee member of a variety of conferences,  such as ICC, VTC, and GLOBECOM.
\end{IEEEbiography}

\begin{IEEEbiography}[{\includegraphics[width=1in,height=1.25in,clip,keepaspectratio]{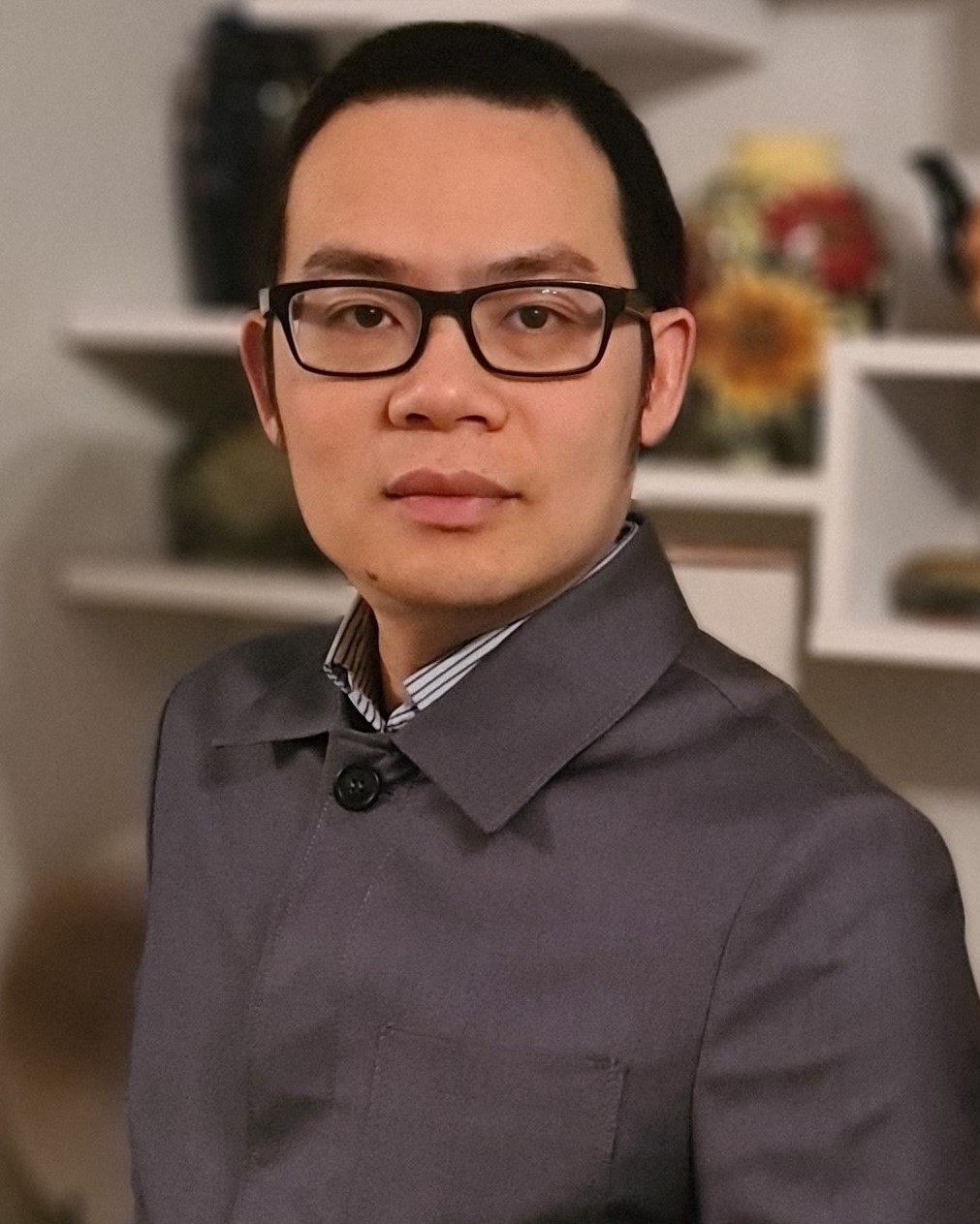}}]
{Hien Quoc Ngo}~(Senior Member, IEEE)~received the B.S. degree in electrical engineering from the Ho Chi Minh City University of Technology, Vietnam, in 2007, the M.S. degree in electronics and radio engineering from Kyung Hee University, South Korea, in 2010, and the Ph.D. degree in communication systems from Link\"oping University (LiU), Sweden, in 2015. In 2014, he visited the Nokia Bell Labs, Murray Hill, New Jersey, USA. From January 2016 to April 2017, Hien Quoc Ngo was a VR researcher at the Department of Electrical Engineering (ISY), LiU. He was also a Visiting Research Fellow at the School of Electronics, Electrical Engineering and Computer Science, Queen's University Belfast, UK, funded by the Swedish Research Council.

Hien Quoc Ngo is currently a Reader (Associate Professor) at Queen's University Belfast, UK. His main research interests include massive MIMO systems, cell-free massive MIMO, physical layer security, and cooperative communications. He has co-authored many research papers in wireless communications and co-authored the Cambridge University Press textbook \emph{Fundamentals of Massive MIMO} (2016).

Dr. Hien Quoc Ngo received the IEEE ComSoc Stephen O. Rice Prize in Communications Theory in 2015, the IEEE ComSoc Leonard G. Abraham Prize in 2017, and the Best PhD Award from EURASIP in 2018. He also received the IEEE Sweden VT-COM-IT Joint Chapter Best Student Journal Paper Award in 2015. He was awarded the UKRI Future Leaders Fellowship in 2019.
Dr. Hien Quoc Ngo currently serves as an Editor for the IEEE Transactions on Communications, the IEEE Transactions on Wireless Communications, the IEEE Wireless Communications Letters, the Digital Signal Processing, and the Elsevier Physical Communication. He was a Guest Editor of IET Communications, special issue on ``Recent Advances on 5G Communications'' and a Guest Editor of  IEEE Access, special issue on ``Modelling, Analysis, and Design of 5G Ultra-Dense Networks'', in 2017. He has been a member of Technical Program Committees for many IEEE conferences such as ICC, GLOBECOM, WCNC, and VTC.
\end{IEEEbiography}

\begin{IEEEbiography}[{\includegraphics[width=1in,height=1.25in,clip,keepaspectratio]{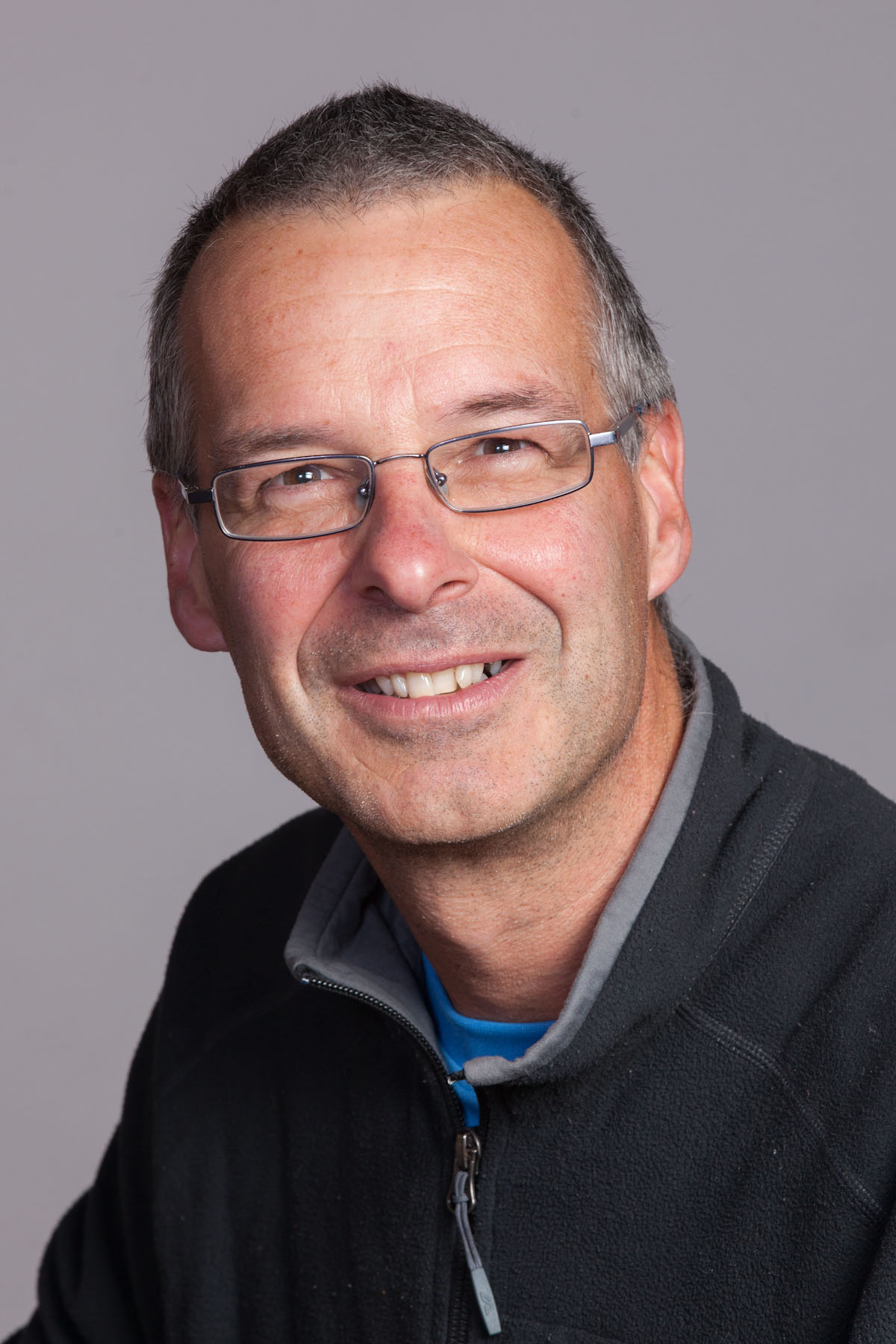}}]
{Peter Smith}~(M'93–SM'01-F'15) received the B.Sc degree in Mathematics and the Ph.D degree in Statistics from the University of London, London, U.K., in 1983 and 1988, respectively. From 1983 to 1986 he was with the Telecommunications Laboratories at GEC Hirst Research Centre. From 1988 to 2001 he was a lecturer in statistics at Victoria University of Wellington, New Zealand. From 2001-2015 he worked in Electrical and Computer Engineering at the University of Canterbury. In 2015 he joined Victoria University of Wellington as Professor of Statistics. He is also an Adjunct Professor in Electrical and Computer Engineering at the University of Canterbury, New Zealand and an Honorary Professor in the School of Electronics, Electrical Engineering and Computer Science,  Queen’s University Belfast. He was elected a Fellow of the IEEE in 2015 and in 2017 was awarded a Distinguished Visiting Fellowship by the UK based Royal Academy of Engineering at  Queen’s University Belfast. In
2018-2019 he was awarded Visiting Fellowships at the University of Bologna, the University of Bristol and the University of Melbourne. His research interests include the statistical aspects of design, modeling and analysis for communication systems, especially antenna arrays, MIMO, cognitive radio, massive MIMO, mmWave systems, reconfigurable intelligent surfaces and the fusion of radar sensing and communications.
\end{IEEEbiography}

\begin{IEEEbiography}[{\includegraphics[width=1in,height=1.25in,clip,keepaspectratio]{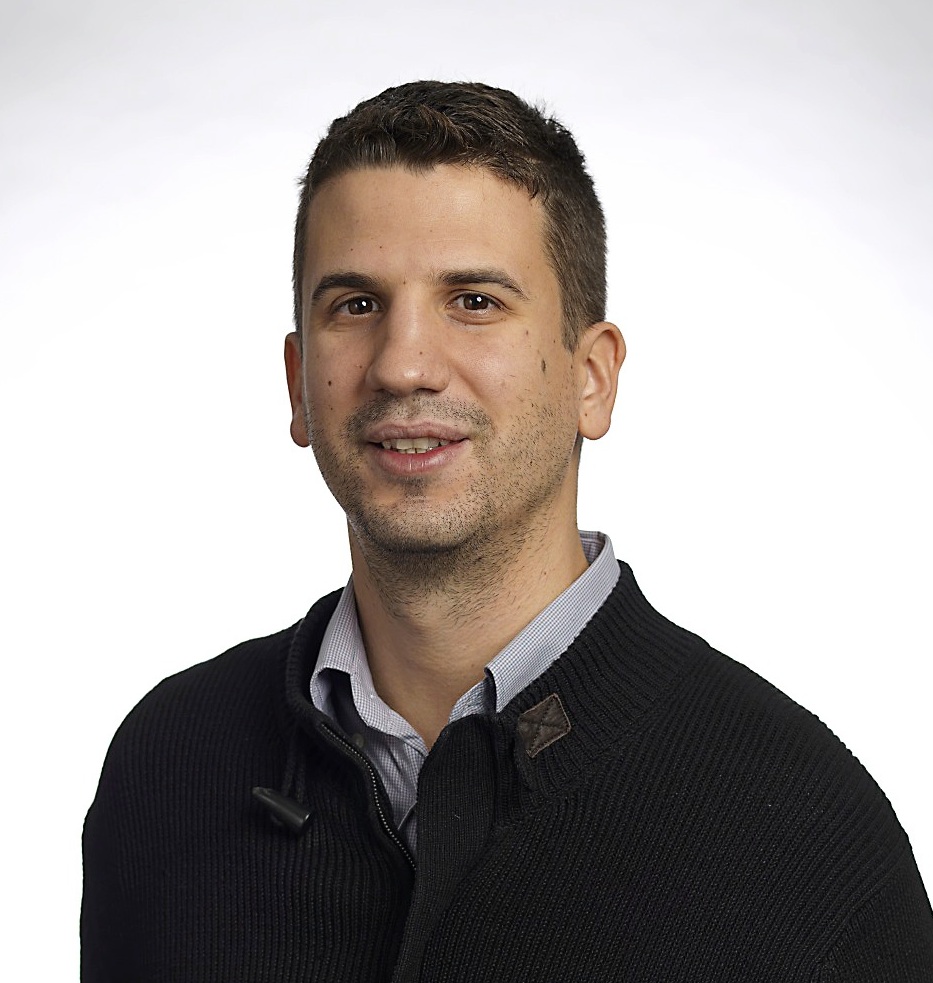}}]
{Michail Matthaiou}~(Fellow, IEEE) was born in Thessaloniki, Greece in 1981. He obtained the Diploma degree (5 years) in Electrical and Computer Engineering from the Aristotle University of Thessaloniki, Greece in 2004. He then received the M.Sc. (with distinction) in Communication Systems and Signal Processing from the University of Bristol, U.K. and Ph.D. degrees from the University of Edinburgh, U.K. in 2005 and 2008, respectively. From September 2008 through May 2010, he was with the Institute for Circuit Theory and Signal Processing, Munich University of Technology (TUM), Germany working as a Postdoctoral Research Associate. He is currently a Professor of Communications Engineering and Signal Processing and Deputy Director of the Centre for Wireless Innovation (CWI) at Queen’s University Belfast, U.K. after holding an Assistant Professor position at Chalmers University of Technology, Sweden. His research interests span signal processing for wireless communications, beyond massive MIMO, intelligent reflecting surfaces, mm-wave/THz systems and deep learning for communications.

 Dr. Matthaiou and his coauthors received the IEEE Communications Society (ComSoc) Leonard G. Abraham Prize in 2017. He currently holds the ERC Consolidator Grant BEATRICE (2021-2026) focused on the interface between information and electromagnetic theories. To date, he has received the prestigious 2023 Argo Network Innovation Award, the 2019 EURASIP Early Career Award and the 2018/2019 Royal Academy of Engineering/The Leverhulme Trust Senior Research Fellowship. His team was also the Grand Winner of the 2019 Mobile World Congress Challenge. He was the recipient of the 2011 IEEE ComSoc Best Young Researcher Award for the Europe, Middle East and Africa Region and a co-recipient of the 2006 IEEE Communications Chapter Project Prize for the best M.Sc. dissertation in the area of communications. He has co-authored papers that received best paper awards at the 2018 IEEE WCSP and 2014 IEEE ICC. In 2014, he received the Research Fund for International Young Scientists from the National Natural Science Foundation of China. He is currently the Editor-in-Chief of Elsevier Physical Communication, a Senior Editor for \textsc{IEEE Wireless Communications Letters} and \textsc{IEEE Signal Processing Magazine}, and an Associate Editor for \textsc{IEEE Transactions on Communications}. He is an IEEE Fellow.

\end{IEEEbiography}

\end{document}